\documentclass[sn-mathphys-num]{sn-jnl}
\usepackage{fix-cm}

\usepackage{graphicx}%
\usepackage{multirow}%
\usepackage{amsmath,amssymb,amsfonts}%
\usepackage{amsthm}%
\usepackage{mathrsfs}%
\usepackage[title]{appendix}%
\usepackage{xcolor}%
\usepackage{textcomp}%
\usepackage{manyfoot}%
\usepackage{booktabs}%
\usepackage{algorithm}%
\usepackage{algorithmicx}%
\usepackage{algpseudocode}%
\usepackage{listings}%
\usepackage{fancyhdr}
\usepackage{natbib}
\usepackage{bm}

\newtheorem{theorem}{Theorem}
\renewcommand{\footnoterule}{%
	
	\hrule width 0.3\textwidth height 0.4pt
	\kern 6pt  
}

\newtheorem{lemma}{Lemma}%

\newtheorem{definition}{Definition}%
\usepackage{hyperref}
\begin{document}
	
	\title[Article Title]{Linearly Homomorphic Ring Signature Scheme over Lattices}

	\author[1,2]{\fnm{Heng} \sur{Guo}}\email{guoheng@ruc.edu.cn}
	\equalcont{These authors contributed equally to this work.}

	\author[1]{\fnm{Jia} \sur{Li}}\email{lijia1230@ruc.edu.cn}
	\equalcont{These authors contributed equally to this work.}

     \author[3]{\fnm{Yanan} \sur{Wang}}\email{math\_wyn@buaa.edu.cn}
	\equalcont{These authors contributed equally to this work.}

     \author[4]{\fnm{Fengxia} \sur{Liu}}\email{shunliliu@gbu.edu.cn}
	\equalcont{These authors contributed equally to this work.}

 \author[4]{\fnm{Zhiyong} \sur{Zheng}}\email{zhengzy@ruc.edu.cn}
	\equalcont{These authors contributed equally to this work.}

	\author*[1]{\fnm{Kun} \sur{Tian}}\email{tkun19891208@ruc.edu.cn}
	\equalcont{These authors contributed equally to this work.}

	\affil[1]{\orgdiv{School of Mathematics}, \orgname{Renmin University of China}, \orgaddress{ \city{Beijing},  \country{China}}}
	
	\affil[2]{\orgdiv{Institute of Interdisciplinary Studies}, \orgname{Renmin University of China}, \orgaddress{ \city{Beijing},  \country{China}}}
	
     \affil[3]{\orgdiv{School of Mathematical Sciences}, \orgname{Beihang University}, \orgaddress{ \city{Beijing}, \country{China}}}

	\affil[4]{\orgname{Great Bay University}, \orgaddress{ \city{Dongguan}, \state{Guangdong Province}, \country{China}}}

	
	\abstract{Construct the first provably secure linear homomorphic ring signature scheme. Ring signatures allow a signer to anonymously sign a message on behalf of a user group (ring) and are widely applied in areas such as identity protection, electronic voting, and privacy enhancement in blockchain. Homomorphic signatures, on the other hand, support verifiable computations on signed data. The integration of anonymity and computability in homomorphic ring signatures holds the potential to create new application scenarios for privacy-preserving distributed systems. It is worth noting that Choi and Kim first introduced the concept of linear homomorphic ring signatures in 2017 and proposed a specific scheme. However, their scheme lacks a complete security proof, leaving its security theoretically unconfirmed. To address this research gap, this paper presents the first provably secure lattice-based linear homomorphic ring signature scheme, designed for scenarios where the ring size is $\mathcal{O}(\log n)$. This scheme not only combines the anonymity of ring signatures with the malleability of homomorphic signatures but also achieves resistance against quantum attacks.
}

	
	
	
	\maketitle
	
	Functioning as a fundamental pillar of public-key cryptography, the digital signature mechanism was formally presented by Diffie and Hellman in their landmark 1976 paper, New Directions in Cryptography \cite{1}. This seminal work catalyzed extensive research into digital signatures. Driven by objectives to enhance efficiency and fortify security, subsequent research has engendered a sophisticated body of theory and an array of practical techniques. Representative works can be found in references \cite{2,3,4,5,6}. 

In 2001, Rivest, Shamir, and Tauman \cite{7,8} first proposed the ring signature scheme. The core concept of this cryptographic primitive is as follows: the signer can select a group of users through their public keys and then sign a message on behalf of this group (referred to as a "ring"). Signing in the name of these users implies that it is impossible to determine which specific user within the ring actually signed the message. Ring signatures exhibit notable flexibility characteristics: the signer can freely and spontaneously form the ring structure without the involvement of a trusted third party; in fact, ring members do not even need to be aware of each other's existence. Ring signature technology has garnered sustained attention in academia due to its potential application value in addressing real-world technical and social challenges. For instance, it protects whistleblower identities \cite{7}, enables participant anonymity in electronic voting systems and digital cash protocols \cite{9}, facilitates transaction anonymization in blockchain applications \cite{10}, and particularly in the Monero cryptocurrency system, effectively achieves transaction untraceability \cite{11}.  Representative ring signature schemes can be referred to in \cite{12,13,14,15,16}.

Homomorphic signatures were initially formalized by Johnson et al. \cite{17} in 2002 as a cryptographic primitive that permits authorized computations on authenticated data while maintaining verifiability without access to the original signing key. This functionality has demonstrated significant utility across various applications, most notably in network coding protocols \cite{18,19,20} and cloud-based outsourced computation systems \cite{21,22,23}. Subsequent research has systematically organized homomorphic signature schemes into three principal classes: linearly homomorphic schemes supporting additive operations \cite{24,25,26,27,28}, polynomially homomorphic schemes enabling polynomial function evaluation \cite{29,30,31,32}, and fully homomorphic schemes allowing arbitrary computations \cite{33,34,35}.

 Although ring signatures and homomorphic signatures each provide robust cryptographic functionalities, the fusion of these two properties naturally gives rise to a novel cryptographic primitive: homomorphic ring signatures. This primitive can simultaneously offer the anonymity guarantees of ring signatures and the computational malleability of homomorphic signatures, enabling computations on authenticated data within dynamically formed rings while preserving signer privacy. Such schemes hold the potential to unlock new applications in privacy-preserving distributed systems, such as confidential blockchain transactions requiring signed input aggregation, and secure multi-party computation scenarios where participants demand both anonymity and verifiable correctness of results. It is worth noting that Choi and Kim, in their 2017 publication, first introduced the concept of linear homomorphic ring signatures and proposed a concrete scheme \cite{41, 42}. However, their scheme did not provide a complete security proof, leaving its security theoretically unconfirmed. 

	\textbf{Our contributions.} 
	To fill the aforementioned research gap, this paper proposes the first provably secure lattice-based linear homomorphic ring signature scheme, suitable for scenarios with a ring size of $\mathcal{O}(\log n)$. This scheme not only preserves the essential properties of ring signatures and homomorphic signatures—namely, anonymity and computational practicality—but also achieves the additional security advantage of resistance against quantum attacks.
	
	\textbf{Organization.}
The structure of this paper is as follows: Section 2 introduces the basic notation definitions and theoretical foundations. Section 3 presents the definition of the signature scheme and its security model. Section 4 elaborates on the construction of the scheme and provides a formal correctness proof. Section 5 demonstrates that the proposed scheme satisfies unconditional anonymity and existential unforgeability. Finally, Section 6 concludes the paper and discusses open problems.

\section{Preliminaries}

\subsection{\textbf{Notation}}

\noindent \textbf{Mathematical Notation:} Vectors are denoted by bold lowercase letters (e.g., $\boldsymbol{u}$, $\boldsymbol{v}$), and matrices are denoted by uppercase letters (e.g., $\mathbf{A}$, $\mathbf{T}$). The horizontal concatenation of matrices is denoted as $[\mathbf{A} \mid \mathbf{T}]$ or $[\mathbf{A}\parallel \mathbf{T}]$. The ring of integers is denoted by $\mathbb{Z}$, the ring of integers modulo $q$ by $\mathbb{Z}_q$ (or $\mathbb{Z}/q\mathbb{Z}$), a general field by $\mathbb{F}$, and the finite field with $q$ elements (where $q$ is a prime or a prime power) by $\mathbb{F}_q$. The $n$-dimensional Euclidean space is denoted by $\mathbb{R}^n$. 

	\vspace{1\baselineskip}
\noindent \textbf{Matrix Norms and Orthogonalization:} For a matrix $\mathbf{A}=[\mathbf{a}_{1},\dots,\mathbf{a}_{m}]\in \mathbb{R}^{n\times m}$, the matrix norm is defined as $\|\mathbf{A}\|=\max_{1\leq i\leq m}\|\mathbf{a}_{i}\|$, where $\|\mathbf{a}_{i}\|$ represents the $\ell_2$-norm of the vector $\mathbf{a}_{i}$. Furthermore, let $\widetilde{\mathbf{A}}=[\widetilde{\mathbf{a}}_{1},\dots,\widetilde{\mathbf{a}}_{n}]$ denote the Gram-Schmidt orthogonalization of $\mathbf{A}$, i.e.,
\[
\widetilde{\mathbf{a}}_{1}= \mathbf{a}_{1}, \quad \widetilde{\mathbf{a}}_{i}= \mathbf{a}_{i}-\sum_{j=1}^{i-1}\frac{\langle\mathbf{a}_{i}, \widetilde{\mathbf{a}}_{j}\rangle}{\langle \widetilde{\mathbf{a}}_{j}, \widetilde{\mathbf{a}}_{j}\rangle} \widetilde{\mathbf{a}}_{j}, \quad 2\leq i\leq n,
\]
where $\langle\cdot,\cdot\rangle$ denotes the standard inner product in Euclidean space.

	\vspace{1\baselineskip}
\noindent \textbf{Asymptotic Notation:}
	For asymptotic analysis, we employ the following notation:
		\begin{itemize}
			\item $f(n) = \mathcal{O}(g(n))$ if $\exists c>0, N\in\mathbb{N}$ such that $\forall n>N$, $f(n) \leq cg(n)$
			\item $f(n) = \widetilde{\mathcal{O}}(g(n))$ when $f(n) = \mathcal{O}(g(n)\log^{c'}n)$ for some $c'>0$
			\item $f(n) = \omega(g(n))$ if $\forall c>0$, $\exists N$ such that $g(n) \leq cf(n)$ whenever $n > N$
		\end{itemize}
		
		Let $n$ be the security parameter. A function $f(n)$ is said to be \textit{polynomial} (written $f(n) = \mathrm{poly}(n)$) if there exists $c>0$ such that $f(n) = \mathcal{O}(n^c)$; conversely, it is \textit{negligible} (denoted $\mathrm{negl}(n)$) if for every $c>0$ we have $f(n) = \mathcal{O}(n^{-c})$. An event is said to occur with \textit{overwhelming probability} when its probability satisfies $\Pr[\text{event}] \geq 1 - \mathrm{negl}(n)$.

	\vspace{1\baselineskip}
\noindent \textbf{Probabilistic Notation and Algorithms:}
For any given distribution $\mathcal{D}$, the expression $x \sim \mathcal{D}$ signifies that the variable $x$ follows the distribution $\mathcal{D}$, while $x \leftarrow \mathcal{D}$ denotes the process of sampling a random value from $\mathcal{D}$. When considering a finite set $\mathcal{X}$, the notation $x \stackrel{\$}{\leftarrow} \mathcal{X}$ indicates that $x$ is chosen uniformly at random from $\mathcal{X}$. Regarding probabilistic polynomial-time (PPT) algorithms $\mathsf{Alg}$, the notation $y \leftarrow \mathsf{Alg}(x)$ represents the procedure of generating an output $y$ by executing the algorithm on input $x$.

	\vspace{1\baselineskip}
\subsection{\textbf{Statistical Distance, Entropy and Decomposition Algorithm}}	

\begin{definition}\label{d3.45}(Statistical distance, \cite{24})
Let \( M \subset \mathbb{R}^{n} \) be a finite or countable set, and let \( X \) and \( Y \) be discrete random variables taking values in \( M \). The statistical distance \( X \) and \( Y \) is defined as:
\[
\Delta(X, Y) = \frac{1}{2} \sum_{a \in M} \left| \mathbb{P}\{X = a\} - \mathbb{P}\{Y = a\} \right|.
\]
\end{definition}

Two probability distributions $\mathcal{D}_0$ and $\mathcal{D}_1$ are said to be \emph{statistically indistinguishable} if their statistical distance 
\[
\Delta(\mathcal{D}_0, \mathcal{D}_1)
\] 
is a negligible function in the security parameter $n$.

	\vspace{1\baselineskip}
\begin{definition}\label{d3.2}( \emph{Min-entropy}\cite{33})
	For a random variable \( X \), its min-entropy is defined as:
	\[
	H_{\infty}(X) = -\log\left( \max_{x \in X} \Pr\{X = x\} \right).
	\]
	
	The average conditional min-entropy of a random variable \( X \) conditional on a correlated variable \( Y \) is defined as:
	\[
	H_{\infty}(X|Y) = -\log\left( \mathbb{E}_{y \in Y} \left\{ \max_{x \in X} \Pr\{X = x | Y = y\} \right\} \right).
	\]
\end{definition}

The optimal probability of an unbounded adversary guessing $X$ given the correlated value $Y$ is $2^{-H_{\infty}(X|Y)}$\cite{33}.

\vspace{1\baselineskip}
\begin{lemma}\label{l3.1}(\cite{33})
	Let $X$, $Y$ be arbitrarily random variables where the support of $Y$ lies in  $\mathcal{Y}$, Then
	$$ H_{\infty}(X|Y)> H_{\infty}(X)-\log(|\mathcal{Y}|). $$
\end{lemma}		

The following decomposition algorithm is derived from \cite{24}. This algorithm decomposes a vector $\boldsymbol{w} \in \mathbb{F}_{2}^{2k}$ into a pair of vectors $(\boldsymbol{u}, \boldsymbol{v})$ belonging to a specific set, as described below:

\begin{theorem}\label{de}( \emph{Decomposition Algorithm}\cite{24})
Let \( k \) be an odd integer with \( k = \mathrm{poly}(n) \). Define the sets
\[
\begin{aligned}
	\mathcal{X} &= \big\{ x \in \mathbb{F}_{2}^{2k} \mid \|x\| = k-1 \big\} \sqcup \{0\}, \\
	\mathcal{Y} &= \big\{ y \in \mathbb{F}_{2}^{2k} \mid \|y\| = k \big\} \sqcup \{0\}, \\
	\mathcal{Z} &= \big\{ z \in \mathbb{F}_{2}^{2k} \mid \|z\| = k+1 \big\} \sqcup \{0\}.
\end{aligned}
\]
There exists a polynomial-time algorithm \(\textsf{Decompose}(\boldsymbol{w})\) that, given an input \(\mathbf{w} \in \mathbb{F}_{2}^{2k}\), outputs a pair of vectors \((\mathbf{u}, \mathbf{v})\) such that:
\begin{itemize}
	\item \(\boldsymbol{u} \in \boldsymbol{X}\) and \(\boldsymbol{v} \in \mathcal{Y}\) or  \(\mathcal{Z}\);
	\item if \(\boldsymbol{w} \neq 0\), then at most one of \(\boldsymbol{u}\) and \(\boldsymbol{v}\) is zero.
\end{itemize}
	\end{theorem}	

\subsection{ \textbf{Background on Lattices and Hard Problems}}

\begin{definition}\label{d3.3}(\emph{Lattice }\cite{36})
	Let $\Lambda\subset \mathbb{R}^{n}$ be a non-empty subset. $\Lambda$ is called a lattice if:
	
	(1) it is an additive subgroup of $\mathbb{R}^{n}$;
	
	(2) there exists a positive constant $\lambda=\lambda(\Lambda)>0$ such that
	$$\mathrm{ min}\{\parallel \boldsymbol{x}\parallel | \boldsymbol{x}\in\Lambda , \boldsymbol{x}\neq 0 \}=\lambda.$$  $\lambda$ is called the minimum distance.
\end{definition}

A full-rank lattice $\Lambda \subseteq \mathbb{R}^n$ can be equivalently characterized by a basis $\mathbf{B}= \{\boldsymbol{b}_1, \dots, \boldsymbol{b}_n\}$, where each $\boldsymbol{b}_i \in \mathbb{R}^n$ is linearly independent. The lattice generated by $\mathbf{B}$ comprises all integer linear combinations of the basis vectors:

\begin{equation*}
	\Lambda(\mathbf{B}) = \left\{ \sum_{i=1}^n c_i\boldsymbol{b}_i \;\Big|\; \boldsymbol{c} = (c_1, \dots, c_n)^\top \in \mathbb{Z}^n \right\}.
\end{equation*}

The dual lattice $\Lambda^*$, which is intrinsically associated with $\Lambda$, is defined as the set of vectors satisfying certain orthogonality conditions with respect to $\Lambda$:

\begin{equation*}
	\Lambda^* = \left\{ \boldsymbol{v} \in \mathbb{R}^n \;\Big|\; \langle \boldsymbol{v}, \boldsymbol{u} \rangle \in \mathbb{Z} \text{ for all } \boldsymbol{u} \in \Lambda \right\}.
\end{equation*}

	\vspace{1\baselineskip}
\begin{definition}\label{d3.4}(\emph{$q$-ary lattices}\cite{36})
	Let $\mathbf{A}\in \mathbb{Z}_{q}^{n\times m}$, $\boldsymbol{u}\in \mathbb{Z}^{n}$. The following two $q$-ary lattices are defined as:
	
	(1)$\Lambda_{q}^{\bot}=\{ \boldsymbol{x}\in \mathbb{Z}^{m} | \mathbf{A}\cdot \boldsymbol{x} \equiv 0 (\mathrm{mod} q)   \}$;
	
	(2)$\Lambda_{q}^{\mathbf{u}}=\{ \boldsymbol{y}\in \mathbb{Z}^{m} | \mathbf{A}\cdot \mathbf{y} \equiv\boldsymbol{u} (\mathrm{mod} q)   \}$.
	
	The set $\Lambda_{q}^{\mathbf{u}}$ is a coset of  $\Lambda_{q}^{\bot}$ since $\Lambda_{q}^{\mathbf{u}} =\Lambda_{q}^{\bot}+\boldsymbol{x}$ for any $\boldsymbol{x}$ such that $\mathbf{A}\cdot\boldsymbol{x}\equiv \boldsymbol{u} (\mathrm{mod}q)$.
\end{definition}

	\vspace{1\baselineskip}
\begin{definition}\label{d3.5}(\emph{Short integer solution}\cite{36})
	Let $n$, $m$, $q$ be positive integers, with $m=\mathrm{poly(}n)$. Let $\mathbf{A}\in \mathbb{Z}_{q}^{n\times m}$ be a uniformly distributed random matrix over $\mathbb{Z}_{q}$, and let $\beta\in \mathbb{R}$ such that $0<\beta<q$. The $\mathrm{SIS}$ problem is to find a short integer solution $\boldsymbol{x}$ satisfying the following condition:
	$$ \mathbf{A}\cdot\boldsymbol{x}\equiv 0(\mathrm{mod}q), \quad \mathrm{and}\quad \boldsymbol{x}\neq0, \parallel \boldsymbol{x}\parallel\leq\beta.$$
	We write the above $\mathrm{SIS}$ problem as $\mathrm{SIS}_{q,n,m,\beta}$ or $\mathrm{SIS}_{q,\beta}$.
\end{definition}	

	\vspace{1\baselineskip}
\begin{theorem}\label{t3.1}(\emph{Worst-case to average-case reduction }\cite{37})
	For any polynomial bounded \(m = \text{poly}(n)\), and any \(\beta > 0\), if \(q \geq \beta \cdot \omega(\sqrt{n\log n})\), then solving the average-case problem  $\mathrm{SIS}_{q,\beta}$  is at least as hard as solving the worst-case problem \(\mathrm{SIVP}{\gamma}\) on any \(n\)-dimensional lattice for \(\gamma=\beta \cdot \tilde{O}(\sqrt{n})\).
\end{theorem}

\subsection{\textbf{Gaussian Distribution and Its Related Algorithms}}	

\begin{definition}\label{d3.6}( \emph{Discrete Gaussian distributions}\cite{24})
	Let $s$ be a positive real number and $\boldsymbol{c}\in \mathbb{R}^{n}$ be a vector. The Gaussian function centered at $\boldsymbol{c}$ with parameter $s$ is defined as: $\rho_{s,\boldsymbol{c}}(\boldsymbol{x})=e^{\frac{-\pi}{s^{2}}\parallel\boldsymbol{x}-\boldsymbol{c}\parallel^{2}}.$ The discrete Gaussian measure $\mathcal{D}_{\Lambda,s, \boldsymbol{c}}$ defined on the lattice $\Lambda$ is given by:
	$$  \mathcal{D}_{\Lambda,s, \boldsymbol{c}}=\frac{\rho_{s,\boldsymbol{c}}(\mathbf{x})}{\rho_{s,\\boldsymbol{c}}(\Lambda)},$$
	where $\rho_{s,\boldsymbol{c}}(\Lambda)=\sum_{\boldsymbol{x}\in \Lambda}\rho_{s,\boldsymbol{c}}(\boldsymbol{x})$.
\end{definition}

	\vspace{1\baselineskip}
\begin{theorem}\label{t3.2}(\cite{38})
	Let $q\geq3$ be odd, $n$ be a positive integer,  and let $m:=\lceil 6n\log q\rceil.$  There is a probabilistic
	polynomial-time algorithm \textsf{TrapGen}$(q,n,m)$ that outputs a pair $(\mathbf{A}, \mathbf{T})$ such that $\mathbf{A}$ is statistically close to a uniform rank $n$ matrix in $\mathbb{Z}_{q}^{n\times m}$, and $\mathbf{T}\in\mathbb{Z}^{n\times n}$ is a basis for $\Lambda_{q}^{\bot}(A)$ satisfying
	$$ \parallel\widetilde{\mathbf{T}}\parallel\leq \mathcal{O}(\sqrt{n\log q}) \quad\text{and}\quad \parallel \mathbf{T}\parallel \leq \mathcal{O}(n\log q).$$
\end{theorem}		

	\vspace{1\baselineskip}
In \cite{37}, Gentry and colleagues presented a method for sampling from discrete Gaussian distributions with any short basis.

	\vspace{1\baselineskip}
\begin{lemma}\label{l3.3}(\emph{Sampling from discrete Gaussian }\cite{37})
	Let $q\geq2$ , $\mathbf{A}\in \mathbb{Z}_{q}^{n\times m}$ with $m>n$ and let $\mathbf{T}$ be a basis for $\Lambda_{q}^{\bot}(\mathbf{A})$ and $s\geq \widetilde{\mathbf{T}}\cdot\omega(\sqrt{\log m })$. Then for $\boldsymbol{c}\in \mathbb{R}^{n}$ and $\boldsymbol{u}\in \mathbb{Z}_{q}^{n}$:
	\begin{enumerate}
		\item  There is a  probabilistic polynomial-time algorithm \textsf{SampleGaussian}$(\mathbf{A},\mathbf{T},s,\boldsymbol{c})$ that outputs $\boldsymbol{x}\in \Lambda_{q}^{\bot}(\mathbf{A})$ drawn from a distribution statistically close to $\mathcal{D}_{\Lambda_{q}^{\bot}(\mathbf{A}),s,\boldsymbol{c}}$;
		\item  There is a  probabilistic polynomial-time algorithm \textsf{SamplePre}$(\mathbf{A},\mathbf{T},\boldsymbol{u},s)$ that outputs  $\mathbf{x}\in \Lambda_{q}^{\boldsymbol{u}}(\mathbf{A})$  sampled from a distribution statistically close to  $\mathcal{D}_{\Lambda_{q}^{\boldsymbol{u}}(\mathbf{A}),s}$, whenever $\Lambda_{q}^{\boldsymbol{u}}(\mathbf{A})$ is not empty.
	\end{enumerate}
\end{lemma}		

	\vspace{1\baselineskip}
For any parameter $s$ such that $s \geq \omega(\sqrt{\log n})$, the probabilistic algorithm $\mathsf{SampleDom}(1^n, s)$ produces samples following a discrete Gaussian distribution over $\mathbb{Z}^n$. Given a randomly sampled vector $\boldsymbol{x} \leftarrow \mathsf{SampleDom}(1^n, s)$, its output distribution is statistically close to $\mathcal{D}_{\mathbb{Z}^n, s}$. Additionally, the conditional min-entropy of the generated samples is $\omega(\log n)$ \cite{37}.

	\vspace{1\baselineskip}
\begin{lemma}\label{l3.4}(\cite{37})
Let \( n \) and \( h \) be positive integers, and let \( q \) be a prime satisfying \( n \geq 2h \log q \). Then, for all matrices \( \mathbf{A}\in \mathbb{Z}_q^{h \times n} \) except a negligible fraction of \( 2q^{-h} \), and for any \( s \geq \omega(\sqrt{\log n}) \), the distribution of \( \mathbf{\alpha} = \mathbf{A} \cdot \boldsymbol{x} \pmod{q} \) is statistically indistinguishable from uniform over \( \mathbb{Z}_q^h \), where \( \mathbf{x} \) is sampled as \( \boldsymbol{x} \stackrel{\$}{\leftarrow} \textsf{SampleDom}(1^n, s) \).  
\end{lemma}

	\vspace{1\baselineskip}
The generalized sampling algorithm $\mathsf{GenSamplePre}$ was first proposed in \cite{40}.  \cite{15} extended this technique by selecting different parameters and expanding the lattice structure, obtaining the following improved lattice extension technique that is more suitable for the subsequent ring signature scheme.

	\vspace{1\baselineskip}
\begin{theorem}\label{l3.444}(Sampling Preimage for Extended Lattice\cite{15})
	Let $k,k_{1},k_{2},k_{3},k_{4}$ be positive integers with $k=k_{1}+k_{2}+k_{3}+k_{4}$. Let $S=[k]=\{1,2,\ldots,k\}$ (interpretable as an index set), and denote $\mathbf{A}_{S}=[\mathbf{A}_{S_{1}}|\mathbf{A}_{S_{2}}|\mathbf{A}_{S_{3}}|\mathbf{A}_{S_{4}}]$, where $\mathbf{A}_{S_{1}}\in \mathbb{Z}_{q}^{h\times k_{1}n}$, $\mathbf{A}_{S_{2}}\in \mathbb{Z}_{q}^{h\times k_{2}n}$, $\mathbf{A}_{S_{3}}\in \mathbb{Z}_{q}^{h\times k_{3}n}$, and $\mathbf{A}_{S_{4}}\in \mathbb{Z}_{q}^{h\times k_{4}n}$. Let $\mathbf{A}_{R}=[\mathbf{A}_{S_{1}}|\mathbf{A}_{S_{3}}]\in \mathbb{Z}_{q}^{h\times (k_{1}+k_{3})n}$, given a short basis $\mathbf{T}_{R}$ for the lattice $\Lambda_{q}^{\perp}(\mathbf{A}_{R})$, and a real number $V\geq \|\widetilde{\mathbf{T}}_{R}\|\omega(\sqrt{\log(k_{1}+k_{3})n})$, $\mathbf{y}\in \mathbb{Z}_{q}^{h}$. Then there exists a probabilistic polynomial-time (PPT) algorithm $\mathsf{GenSamplePre}(\mathbf{A}{S}, \mathbf{A}_{R}, \mathbf{T}_{R}, \mathbf{y}, V)$ that outputs a vector $\boldsymbol{e}\in \mathbb{Z}_{q}^{kn}$ whose distribution is statistically indistinguishable from $D_{\Lambda_{q}^{\mathbf{y}}(\mathbf{A}_{S}), V}$.
\end{theorem}

	\vspace{1\baselineskip}
\begin{definition}\label{d3.7}( \emph{Smoothing parameter}\cite{39})
	For any $n$-dimensional lattice $\Lambda$ and any given $\epsilon>0$, the smoothing parameter of the lattice is defined as
	\begin{equation*}
		\eta_{\epsilon}(\Lambda)=\min\left\{s > 0 \mid \rho_{\frac{1}{s}}(\Lambda^{*}) < 1 + \epsilon\right\}.
	\end{equation*}
\end{definition}

	\vspace{1\baselineskip}
The smoothing parameter $\eta_\epsilon(\Lambda_q^\perp(\mathbf{A}))$ is at most $\omega(\sqrt{\log n})$ for some negligible $\epsilon$, and this bound holds for all but an exponentially small fraction of matrices $\mathbf{A}\in \mathbb{Z}_q^{h \times n}$ \cite{37}. The detailed description is as follows:

	\vspace{1\baselineskip}
\begin{lemma}\label{l3.5}(\cite{37}, Lemma 5.3)
	Let $q\geq3$, $h$ and $n$ be positive integers satisfying $n\geq 2h\lg q$. Then there exists a negligible function $\epsilon(n)$ such that
	$\eta_{\epsilon}(\Lambda_{q}^{\bot}(A))<\omega(\sqrt{\log n})$ for all but at most a $q^{-h}$ fraction of $A$ in the $\mathbb{Z}_{q}^{h\times n}$ .
\end{lemma}

	\vspace{1\baselineskip}
The following lemma demonstrates that vectors drawn from a discrete Gaussian distribution are concentrated within specific bounds with high probability.

	\vspace{1\baselineskip}
\begin{lemma}\label{l3.2}(\cite{39})
	Let $\Lambda$ be an $n$-dimensional lattice, and $T$ be a basis of the lattice $\Lambda$. If $s\geq \parallel \widetilde{T}\parallel\cdot \omega(\sqrt{\log n})$, then for any $\mathbf{c}\in \mathbb{R}^{n}$, we have:
	$$ \mathrm{Pr}\{\parallel \boldsymbol{x}-\boldsymbol{c}\parallel>s\sqrt{n}:\boldsymbol{x}{\leftarrow} \mathcal{D}_{\Lambda,s,\boldsymbol{c}}\}\leq \mathrm{negl}(n) .$$
\end{lemma}

\subsection{\textbf{Statistical Properties Related to Discrete Gaussian Distributions}}

\begin{theorem}\label{t3.33}(\cite{25})
Let	$\boldsymbol{t}_i \in \mathbb{Z}^m$ and $\boldsymbol{x}_i$ are mutually independent random variables sampled from a Gaussian distribution $D_{\boldsymbol{t}_i + \Lambda,\sigma}$ over $\boldsymbol{t}_i + \Lambda$ for $i = 1,2,\cdots,k$ in which $\Lambda$ is a lattice and $\sigma \in \mathbb{R}$ is a parameter. Let $\boldsymbol{c} = (c_1,\cdots,c_k) \in \mathbb{Z}^k$ and $g = \mathrm{gcd}(c_1,\cdots,c_k)$, $\boldsymbol{t} = \sum_{i = 1}^{k} c_i \boldsymbol{t}_i$. If $\sigma > \|\boldsymbol{c}\| \eta_{\epsilon}(\Lambda)$ for some negligible number $\epsilon$, then $\boldsymbol{z}= \sum_{i = 1}^{k} c_i \boldsymbol{x}_i$ statistically closes to $D_{\boldsymbol{t} + g\Lambda,\|\boldsymbol{c}\|\sigma}$.
\end{theorem}

	\vspace{1\baselineskip}
\begin{lemma}\label{l3.777}(\cite{24})
	Let \(\mathbf{A}\in \mathbb{Z}_q^{n\times m}, s > 0, \boldsymbol{u} \in \mathbb{Z}^{m}\), and \(\Lambda = \Lambda_q^{\boldsymbol{u}}(A)\). If \(\boldsymbol{x}\) is sampled from \(\mathcal{D}_{\mathbb{Z}^m, s}\) conditioned on \(\mathbf{A}\boldsymbol{x} \equiv \boldsymbol{u} \pmod{q}\), then the conditional distribution of \(\boldsymbol{x}\) is \(\mathcal{D}_{\Lambda,s}\).
\end{lemma}

	\vspace{1\baselineskip}
\begin{lemma}\label{gh}(\cite{24})
	Let $k = \text{poly}(n)$ be even, $q = \text{poly}(n) \geq (nk)^2$, and $V = \sqrt{2nk\log q}\log n \geq \omega(\sqrt{\log n})$. Then $\mathcal{D}_{\mathbb{Z}^{n}, V}$ and $\mathcal{D}_{\mathbb{Z}^{n}, \sqrt{\frac{k \pm 2}{k}}V}$ are statistically indistinguishable.
\end{lemma}

\section{Definition and Security Model of LHRS}

\subsection{ \textbf{Definition of LHRS}}

	\begin{definition}\label{2.3}(LHRS scheme)
The linearly homomorphic ring signature is composed of a set of probabilistic polynomial-time (PPT) algorithms $\mathcal{LHRS}$
=(\textsf{Setup},\textsf{HR-Sign},\textsf{Combine},\textsf{HR-Verify}), which are defined as follows:
	
	\begin{enumerate}
	\item[$\bullet$] \textsf{Setup} ($1^{\lambda}$, \textsf{pp}): Given a security parameter $\lambda$ and public parameters \textsf{pp},  this algorithm outputs a key pair $(pk, sk)$. Here, the public parameters \textsf{pp} determine the tag space $\mathcal{T}$ , the message space $\mathcal{M}$,  the signature space $\Sigma$, and and the maximum number of homomorphic operations $k_{0}$.

		\item[$\bullet$] \textsf{HR-Sign}$(\textsf{pk}_{s},\textsf{sk}_{s}, R, \tau, \mathbf{m})$: Given the key pair $(\textsf{pk}_{s},\textsf{sk}_{s})$ of user $s$, the ring’s public key set $R$ (where $\textsf{pk}_{s}\in R$), the dataset label $\tau\in\mathcal{T}$, and the message $\mathbf{m}\in\mathcal{M} $, this algorithm outputs a ring signature $\sigma$ on message $\mathbf{m}$ under label $\tau$, signed using $\textsf{sk}_{s}$.

		\item[$\bullet$] \textsf{Combine}$(R, \tau, \{(c_{j}, \sigma_{j})\}_{j=1}^{\ell})$:  Given the ring $R$, the tag $\tau$, and a set of tuples $\{(c_{j}, \sigma_{j})\}_{j=1}^{\ell})$ (where $\ell\leq k_{0}$), this algorithm outputs a ring  signature $\sigma$ on the message $\sum_{j=1}^{\ell}c_{j}\mathbf{m}_{j}$ under the label tag $\tau$ for   the ring $R$

		\item[$\bullet$] \textsf{HR-Verify}$(R,\tau,\mathbf{m}, \sigma)$:  Given the  ring $R$, the tag $\tau$,  the  message $\mathbf{m}$, and the ring signature $\sigma$, the algorithm outputs 0 (reject) or 1 (accept).
		
	\end{enumerate}
	
\end{definition}

\textbf{Correctness}  requires that: 

(1) For all valid tags 
$\tau\in\mathcal{T} $ and messages $\mathbf{m}\in \mathcal{M}$,  if a ring signature $\sigma$ is generated as $   \sigma\leftarrow\textsf{HR-Sign}(\textsf{pk}_{s},\textsf{sk}_{s}, R, \tau, \mathbf{m})$, then the verification algorithm must satisfy  
$$ 1\leftarrow \textsf{HR-Verify}(R,\tau,\mathbf{m}, \sigma).$$

(2) If $\sigma_{j}\leftarrow\textsf{HR-Sign}(\textsf{pk}_{s},\textsf{sk}_{s}, R, \tau, \mathbf{m}_{j})$, then 
$$  1\leftarrow\textsf{HR-Verify}(R,\tau,\sum_{j=1}^{\ell}c_j\mathbf{m}_{j}, \textsf{Combine}(R, \tau, \{(c_{j}, \sigma_{j})\}_{j=1}^{\ell})).$$

\subsection{ \textbf{Security Model of LHRS}}

In ring signature schemes \cite{13,15}, security typically encompasses two fundamental requirements: anonymity and unforgeability. Similarly, linearly homomorphic ring signature schemes must also satisfy these two properties. Specifically, they should achieve anonymity under full key exposure and unforgeability against insider corruption.  

\vspace{1\baselineskip}
	\begin{definition}\label{2.4}(Anonymity under Full Key Exposure )
The security of the LHRS scheme against anonymity under full key exposure is defined through the following game played between a challenger $\mathcal{B}_{1}$ and a polynomial-time adversary  $\mathcal{A}_{1}$:
	
	\begin{enumerate}
		\item[$\bullet$] \textsf{Setup}: The challenger $\mathcal{B}_1$ runs the $\mathsf{Setup}(1^\lambda, \mathsf{pp})$ algorithm $\ell$ times to generate key pairs $(\textsf{pk}_1, \textsf{sk}_1), \dots, (\textsf{pk}_\ell, \textsf{sk}_\ell)$, where $\ell$ is a game parameter, and sends the public key set $\{\textsf{pk}_s\}_{s=1}^\ell$ to the adversary $\mathcal{A}_1$. 
		
		\item[$\bullet$] \textsf{Queries}: After receiving $\{(\textsf{pk}_s, \textsf{sk}_s)\}_{s=1}^\ell$, the adversary $\mathcal{A}_1$ can make the following two types of queries:
		
		-\textsf{Corruption query:}  The adversary $\mathcal{A}_1$ selects a user $s \in [\ell]=\{1,2,...,\ell\}$ and sends $s$ to the challenger $\mathcal{B}_1$, who responds by returning $\textsf{sk}_s$ to $\mathcal{A}_1$. 
		
		-\textsf{Sign query:} After selecting a user $s$, adversary $\mathcal{A}_2$ adaptively selects a series of message subspaces $V_{si} = \text{span}\{\mathbf{m}_{si}^{(1)}, ..., \mathbf{m}_{si}^{(k_0)}\}$ and queries them to the challenger $\mathcal{B}_2$, for $i = 1, \dots, q_s$, where $q_s$ is the maximum number of allowed queries for user $s$. For each queried subspace $V_{si}$, the challenger $\mathcal{B}_2$ generates a uniformly random tag $\tau_{si} \xleftarrow{\$} \mathcal{T}$.  Then, $\mathcal{B}_2$ computes a signature $\sigma_{si}^{(j)}$ for each message $\mathbf{m}_{si}^{(j)}$ within the subspace, in relation to the tag $\tau_{si}$. Finally, $\mathcal{B}_2$ responds with the tuples $(\mathbf{m}_{si}^{(j)}, \sigma_{si}^{(j)})$ for all $i=1,...,q_s$ and $j=1, \dots, k_0$. 
		
		\item[$\bullet$]  \textsf{Challenge}: Finally, $\mathcal{A}_1$ outputs a challenge tuple $(s_0, s_1, R^*, \tau^*, m^*)$ to $\mathcal{B}_1$, where $s_0$ and $s_1$ are indices satisfying $\textsf{pk}_{s_0} \in R^*$ and $\textsf{pk}_{s_1} \in R^*$, and neither $s_0$ nor  $s_1$ has been corrupted in any prior corruption query. $\mathcal{B}_1$ is then required to generate a signature on $(\tau^*, m^*)$ under ring $R^*$. To do so, $\mathcal{B}_1$  selects a random bit $b \xleftarrow{\$} \{0,1\}$, computes $\sigma_b^* \gets \textsf{HR-Sign}(\textsf{pk}_{s_b}, \textsf{sk}_{s_b}, R^*, \tau^*, m^*)$, and returns $\sigma_b^*$ to $\mathcal{A}_1$. The adversary outputs a guess $b' \in \{0,1\}$ and wins the game if $b' = b$.
	\end{enumerate}
	
The scheme is said to satisfy anonymity  if for every probabilistic polynomial-time adversary $\mathcal{A}_1$, the advantage $\mathbf{Adv}_{LHRS}^{\mathrm{ANON}}(\mathcal{A}_1) = \left|\Pr[b' = b] - \frac{1}{2}\right|$ is negligible in the security parameter.
\end{definition}

\vspace{1\baselineskip}
	\begin{definition}\label{2.44}(Unforgeability against Insider Corruption)
	The security of the LHRS scheme against unforgeability under insider corruption is defined through the following game played between a challenger $\mathcal{B}_{2}$ and a polynomial-time adversary  $\mathcal{A}_{2}$:
	
	\begin{enumerate}
		\item[$\bullet$] \textsf{Setup}: The challenger $\mathcal{B}_2$ runs the $\mathsf{Setup}(1^\lambda)$ algorithm $\ell$ times to generate key pairs $(\textsf{pk}_1, \textsf{sk}_1), \dots, (\textsf{pk}_\ell, \textsf{sk}_\ell)$, where $\ell$ is a game parameter. Let $L=\{\textsf{pk}_s\}_{s=1}^\ell$ and sends the public key set $L$ to the adversary $\mathcal{A}_2$ 
		
		\item[$\bullet$] \textsf{Queries}: After receiving $L=\{\textsf{pk}_s\}_{s=1}^\ell$, the adversary $\mathcal{A}_2$ can make the following two types of queries:
		
		-\textsf{Corruption query:}  The adversary $\mathcal{A}_2$ selects a user $s \in [\ell]=\{1,2,...,\ell\}$ and sends $s$ to the challenger $\mathcal{B}_2$, who responds by returning $\textsf{sk}_s$ to $\mathcal{A}_2$. 
		
		-\textsf{Signing query:} After selecting a user $s$, adversary $\mathcal{A}_2$ adaptively selects a series of message subspaces $V_{si} = \text{span}\{\mathbf{m}_{si}^{(1)}, ..., \mathbf{m}_{si}^{(k_0)}\}$ and queries them to the challenger $\mathcal{B}_2$, for $i = 1, \dots, q_s$, where $q_s$ is the maximum number of allowed queries for user $s$. For each queried subspace $V_{si}$, the challenger $\mathcal{B}_2$ generates a uniformly random tag $\tau_{si} \xleftarrow{\$} \mathcal{T}$.  Then, $\mathcal{B}_2$ computes a signature $\sigma_{si}^{(j)}$ for each message $\mathbf{m}_{si}^{(j)}$ within the subspace, in relation to the tag $\tau_{si}$. Finally, $\mathcal{B}_2$ responds with the tuples $(\mathbf{m}_{si}^{(j)}, \sigma_{si}^{(j)})$for all $i=1,...,q_s$ and $j=1, \dots, k_0$.

		\item[$\bullet$]  \textsf{Output}: Finally, $\mathcal{A}_2$ outputs a tuple $(R^{*}, \tau^{*}, \mathbf{m}^{*}, \sigma^{*})$. The adversary $\mathcal{A}_2$ wins the game if and only if all the following conditions are satisfied:
			\begin{enumerate}
				\item[1.]  $1\leftarrow \textsf{HR-Verify}(R^{*},\tau^{*},\mathbf{m}^{*}, \sigma^{*});$
				\item[2.] $R^{*}\subset L\setminus C$, where $\mathcal{C}$ denotes the set of all corrupted users (i.e., users whose private keys have been obtained by the adversary);
				\item[3.] No valid signature for the tuple $(R^{*}, \tau^{*}, \mathbf{m}^{*})$ was generated during the signing queries in the game;
				\item[4.]  For any pair $(s,i)$, either $\tau^{*} \neq \tau_{si}$ (Type-1 forgery) where $s \in [\ell]$ and $i \in [q_{s}]$, or there exists a pair $(s,i)$ such that $\tau^{*} = \tau_{si}$ (Type-2 forgery) but $\mathbf{m}^{*} \notin V_{si}$.
			\end{enumerate}
	\end{enumerate}
	The adversary $\mathcal{A}_{2}$'s advantage in the above game is defined as:
$\text{Adv}_{\text{LHRS}}^{\text{UNF}}(\mathcal{A}_{2}) = \Pr[\mathcal{A}_{2} \text{ wins}].$
A LHRS scheme is said to be unforgeable if for every probabilistic polynomial-time adversary $\mathcal{A}_{2}$, the advantage $\text{Adv}_{\text{LHRS}}^{\text{UNF}}(\mathcal{A}_{2})$ is negligible.

\end{definition}

\section{The Proposed Scheme}

In this section, we present a lattice-based linearly homomorphic ring signature scheme and prove that the scheme satisfies correctness with overwhelming probability.

\subsection{\textbf{Basic Construction}}

The algorithms $\textsf{TrapGen}$, $\textsf{Decompose}$, and $\textsf{GenSamplePre}$ used in this section are all from Section 2.

\vspace{1\baselineskip}	
Our LHRS scheme $\mathcal{LHRS}$=\{\textsf{Setup},\textsf{HR-Sign},\textsf{Combine}, \textsf{HR-Verify}\} works as follows:

\begin{enumerate}
\item[$\bullet$]  \textsf{Setup} ($1^{\lambda}$, \textsf{pp}): Let $k = \text{poly}(\lambda)$, and let $k_{0} = \frac{k}{2}$ denote the maximum number of homomorphic operations, which is odd. Define $h = \left\lfloor \frac{n}{6 \log q} \right\rfloor \geq k$, $q \geq (nk)^{3}$, and $V = \sqrt{2nk\log q} \cdot \log n$. Let $\mathcal{H}:{0,1}^*\rightarrow \mathbb{Z}_{q}^{h}$ be a hash function, and set the public parameters as $\textsf{pp} = \{q, k_{0}, k, n, h, V, \mathcal{H}\}$. The tag space is $\mathcal{T} = \{0, 1\}^n$, the message space is $\mathcal{M} = \mathbb{F}_2^k$, and the signature space is $\Sigma = \mathbb{Z}^{\ell n}$, where $\ell$ denotes the number of users in the ring $R$.

The algorithm takes as input a security parameter $n$ and the public parameters \textsf{pp}, and generates the key pair for user $s$ as follows:
\begin{enumerate}
	\item [(1)] Compute $(\mathbf{A}_{s}, \mathbf{T}_{s})\leftarrow\textsf{TrapGen}(q,h,n)$;

	\item [(2)] Set the public key as \( \textsf{pk}_s = \mathbf{A}_s, \) and the secret key as \( \textsf{sk}_s = \mathbf{T}_s\), then output the key pair \( (\textsf{pk}_s, \textsf{sk}_s)\) for user \( s \).
\end{enumerate}

\vspace{1\baselineskip}	
\item[$\bullet$]  \textsf{HR-Sign}$(\textsf{pk}_{s}, \textsf{sk}_{s},  R_{\ell },  \tau,  \boldsymbol{m})$: Given the public/private key pair $(\textsf{pk}_{s},  \textsf{sk}_{s})$ of user $s$, and a ring $R_{\ell}$ composed of $\ell$ users' public keys, for notational simplicity, let \(R_\ell = \{\mathbf{A}_1, \dots, \mathbf{A}_{\ell}\}\) where \(\mathbf{A}_s \in R_\ell\). For a tag $\tau \in \{0, 1\}^{n}$ and a message $\boldsymbol{m}=(m_{1}, \dots, m_k) \in \mathbb{F}_{2}^{k}$, the signature is generated as follows:

\begin{enumerate}
    \item [(1)] Let $\mathbf{A}_{R_{\ell }}=[\mathbf{A}_{i_1}|\mathbf{A}_{i_2}|\dots|\mathbf{A}_{i_\ell }]$. Define an index tag vector $ L_{R_{\ell}}=(i_{1}, i_{2}, \dots, i_{\ell})$, which specifies the correspondence between the matrix $\mathbf{A}_{R_{\ell}}$ and the sequence of ring members $\{1, 2, \ldots, \ell\}$, representing the permutation order of ring members in $\mathbf{A}_{R_{\ell}}$;
    \item[(2)] Compute $\bm{\alpha}_{i}=\mathcal{H}(\tau|i)$ for $i=1,2,\dots,k$;
    
    \item [(3)] Compute $(\boldsymbol{u}, \boldsymbol{v})\leftarrow\textsf{Decompose}(\boldsymbol{m})$, where $\boldsymbol{m}=\boldsymbol{u}+\boldsymbol{v}$, $\boldsymbol{u}=(u_{1}, \dots, u_{k}) \in \mathcal{X}$, and $\boldsymbol{v}=(v_{1}, \dots, v_{k})\in \mathcal{Y}$ or $\mathcal{Z}$; 
    \item [(4)] If $\boldsymbol{u}$ and $\boldsymbol{v}$ contain a zero vector (Note: By Theorem \ref{de}, at most one of $\boldsymbol{u}$ or $\boldsymbol{v}$ is zero), let $\boldsymbol{t} = \sum_{j=1}^{k}  m_j\bm{\alpha}_j$ and compute 
    \[
    \boldsymbol{e} \leftarrow \mathsf{GenSamplePre}(\mathbf{A}_{R_{\ell}},  \mathbf{A}_s,  \mathbf{T}_s,  \boldsymbol{t},  V);
    \]
    Otherwise, compute $\boldsymbol{t}(\boldsymbol{u}) = \sum_{j=1}^k u_j \bm{\alpha}_j$ and $\boldsymbol{t}(\boldsymbol{v}) = \sum_{j=1}^k v_j \bm{\alpha}_j$, then generate
    \[
    \boldsymbol{e}(\boldsymbol{u}) \leftarrow \mathsf{GenSamplePre}(\mathbf{A}_{R_{\ell}},  \mathbf{A}_s,  \mathbf{T}_s,  \boldsymbol{t}(\boldsymbol{u}),  V),
    \]
    \[
    \boldsymbol{e}(\boldsymbol{v}) \leftarrow \mathsf{GenSamplePre}(\mathbf{A}_{R_{\ell}},  \mathbf{A}_s,  \mathbf{T}_s,  \boldsymbol{t}(\boldsymbol{v}),  V), 
    \]
    and define $\mathbf{e}= \mathbf{e}(\boldsymbol{u}) + \mathbf{e}(\boldsymbol{v})$. 
    
    \item [(5)] The output ring signature is denoted as $\sigma=( \boldsymbol{e},  L_{R_{\ell}})$.
\end{enumerate}

	\vspace{1\baselineskip}	
\item[$\bullet$] \textsf{Combine}$(R_{\ell},  \tau,  \{(c_{j},  \sigma_{\tau j})\}_{j=1}^{p})$: 
This algorithm takes as input a ring $R_{\ell}$, a tag $\tau$, and a set of tuples $\{(c_{j},  \sigma_{j})\}_{j=1}^{p}$ (where $p \leq k_{0}$ and $c_{j} \in \mathbb{F}_{2}$), where
$$\sigma_{ j} = ( \boldsymbol{e}_{j},  L_{R_{\ell}})\leftarrow\textsf{HR-Sign}(\textsf{pk}_{s},  \textsf{sk}_{s},  R_{\ell},  \tau,  \boldsymbol{m}_{j}).$$   
It outputs a signature for the message $\sum_{j=1}^{p} c_{j}\boldsymbol{m}_{j}$ as
$$\sigma = \left( \sum_{j=1}^{p} c_{j}\boldsymbol{e}_{j},  L_{R_{\ell}}\right). $$

	\vspace{1\baselineskip}	
\item[$\bullet$] \textsf{HR-Verify}$(R_{\ell}, \tau, \boldsymbol{m},  \sigma)$: Given as input a ring \(R_\ell = \{\mathbf{A}_1,  \dots,  \mathbf{A}_{\ell}\}\), a tag $\tau$, a message $\boldsymbol{m}$, and a signature $\sigma=( \boldsymbol{e},  L_{R_{\ell}})$ where $L_{R_{\ell}}=(i_{1}, i_{2}, \dots, i_{\ell})$, the verification process is as follows:

\begin{enumerate}
    \item [(1)] Compute $\mathbf{A}_{R_{\ell}}=[\mathbf{A}_{i_1}|\mathbf{A}_{i_2}|\dots|\mathbf{A}_{i_\ell}]$;  
    \item [(2)] Compute $\bm{\alpha}_{i}=\mathcal{H}(\tau|i)\in \mathbb{Z}_{q}^{h}$ for $i=1,2,\dots,k$, and $\boldsymbol{t}=\sum_{j=1}^{k}m_{j}\bm{\alpha}_{j}$; 
    \item [(3)] Output 1 if both of the following two conditions are satisfied; otherwise, output 0. 
    \begin{enumerate}
        \item[1)] $\parallel \boldsymbol{e}\parallel\leq kV\sqrt{2k\ell n}$; 
        \item[2)] $\mathbf{A}_{R_{\ell}}\cdot \boldsymbol{e}\pmod q =\boldsymbol{t}. $
    \end{enumerate}
\end{enumerate}

	\end{enumerate}

\textbf{Remark}: For the vector decomposition $(\mathbf{u}, \mathbf{v})$ of a message $\mathbf{m}$, when one of the components is a zero vector, the verification of related conclusions becomes trivial; when both components are non-zero vectors, the verification process is representative of the general case. Therefore, to simplify the presentation, in proving the correctness, anonymity, and unforgeability of the scheme, we only need to analyze the case where $\mathbf{m}$ is decomposed into two non-zero vectors.

\subsection{\textbf{Correctness}}

	\begin{theorem}\label{t6.1}
	The above $\mathcal{LHRS}$ scheme  satisfies correctness with overwhelming probability.
\end{theorem}

\begin{proof}	
    The verification is carried out by considering two cases: 
    \begin{enumerate}
        \item[(1)] If $\sigma=(\boldsymbol{e}, L_{R_{\ell}})\leftarrow\textsf{HR-Sign}(\textsf{pk}_{s}, \textsf{sk}_{s},  R_{\ell},  \tau,  \boldsymbol{m})$. Next, we prove that the verification conditions hold with overwhelming probability. 
        
        By Theorem \ref{t3.2}, we have $\parallel \widetilde{\mathbf{T}}_{s}\parallel\leq \mathcal{O}( \sqrt{h\log q})$ with overwhelming probability. Given $V = \sqrt{2nk\log q} \cdot \log Ln$, it follows that
        $$\frac{V}{\parallel\widetilde{\mathbf{T}}_{s}\parallel}\geq \sqrt{\frac{2kn}{h}}\cdot \log n \geq \sqrt{\log n}, $$ 
        which implies $V \geq \parallel \widetilde{\mathbf{T}}_{s}\parallel \omega(\sqrt{\log n})$. Therefore, Theorem \ref{l3.444} and Lemma \ref{l3.2} ensure that $\parallel \mathbf{e}\parallel=\parallel \boldsymbol{e}(\boldsymbol{u})+\boldsymbol{e}(\boldsymbol{v})\parallel\leq 2V \sqrt{\ell n}<kV\sqrt{k\ell n}$ with overwhelming probability. Furthermore, Lemma \ref{l3.3} gives the congruence relation:	  
        \begin{align*}
            \mathbf{A}_{R_{\ell}} \cdot\boldsymbol{e}\pmod{q} 
            &= 	\mathbf{A}_{R_{\ell}} \cdot (\boldsymbol{e}(\boldsymbol{u})+\mathbf{e}(\boldsymbol{v}))\pmod{q} \\
            &= \boldsymbol{t}(\boldsymbol{u})+\boldsymbol{t}(\boldsymbol{v}) \\
            &= \sum_{j = 1}^{k}u_{j}\bm{\alpha}_{j} +\sum_{j = 1}^{k}v_{j}\bm{\alpha}_{j} \\
            &= \sum_{j = 1}^{k}(u_{j}+v_{j})\bm{\alpha}_{j} \\
            &= \sum_{j = 1}^{k}m_{j}\bm{\alpha}_{j} = \boldsymbol{t}. 
        \end{align*}
        Thus, the scheme satisfies correctness for the signature of a single message.

        \item[(2)] Consider a set of tuples $\{(c_{j}, \sigma_{j})\}_{j=1}^{p}$, where each signature is given by
        $$\sigma_{ j}=(\boldsymbol{e}_{j}, L_{R_{\ell}})\leftarrow\textsf{HR-Sign}(\textsf{pk}_{s}, \textsf{sk}_{s},  R_{\ell},  \tau,  \boldsymbol{m}_{j}). $$
        We next prove that the following equality holds:
        $$\textsf{HR-Verify}(R_{\ell}, \tau, \sum_{j = 1}^{p}c_{j}\boldsymbol{m}_{j},  \textsf{Combine}(R_{\ell},  \tau,  \{(c_{j},  \sigma_{ j})\}_{j=1}^{p}))\rightarrow
        1. $$

        First, according to the algorithmic definition of \textsf{Combine}, we obtain the following output:
        $$ \sigma=( \sum_{j=1}^{p}\boldsymbol{e}_{j}, L_{R_{\ell}})\leftarrow\textsf{Combine}(R_{\ell},  \tau,  \{(c_{j},  \sigma_{j})\}_{j=1}^{p}). $$
        By Lemma \ref{l3.2}, we have
        $$\Big\|\sum_{j=1}^{p} c_j \boldsymbol{e}_j \Big\|\leq 2p V \sqrt{\ell n} \leq2 k_{0} V \sqrt{\ell n}\leq kV\sqrt{k\ell n}. $$
        Let $\boldsymbol{m} = (m_{1},  \dots,  m_{k}) = \sum_{j=1}^{p} c_{j} \boldsymbol{m}_{j}$, where $\boldsymbol{m}_{j} = (m_{j1},  \dots,  m_{jk})$. Then
        $$	m_{i} = \sum_{j=1}^{p} c_{j} m_{ji},  \quad \text{for } 1 \leq i \leq k. 	$$ 
        Since
        $$\mathbf{A}_{R_{\ell }} \cdot c_{j}\boldsymbol{e}_{j} \pmod{q} = c_{j} \mathbf{A}_{R_{\ell}}\cdot\boldsymbol{e}_{j}\pmod{q} = c_{j} \boldsymbol{t}_{j},  \quad \text{where } \boldsymbol{t}_{j} = \sum_{i=1}^{k} m_{ji} \bm{\alpha}_{i}, $$
        we get
        \begin{align*}
            \mathbf{A}_{R_{\ell}} \cdot \left( \sum_{j=1}^{p} c_{j} \boldsymbol{e}_{j} \right)\pmod{q} &= \sum_{j=1}^{p} c_{j} \boldsymbol{t}_{j} \\
            &= \sum_{j=1}^{p} c_{j} \left( \sum_{i=1}^{k} m_{ji} \bm{\alpha}_{i} \right) \\
            &= \sum_{i=1}^{k}  \left( \sum_{j=1}^{p} c_{j} m_{ji} \right) \bm{\alpha}_{i}\\ 
            &= \sum_{i=1}^{k} m_{i} \bm{\alpha}_{i} 
            = \boldsymbol{t}. 
        \end{align*}
        Therefore, 
        $$\textsf{HR-Verify}(R_{\ell}, \tau, \sum_{j = 1}^{p}c_{j}\boldsymbol{m}_{j},  \textsf{Combine}(R_{\ell},  \tau,  \{(c_{j},  \sigma_{ j})\}_{j=1}^{p}))\rightarrow
        1. $$
    \end{enumerate}
\end{proof}

\section{Anonymity and Unforgeability}

\subsection{\textbf{Anonymity}}

Next, we prove that the scheme described above satisfies anonymity under full key exposure. This proof establishes unconditional anonymity by demonstrating that challenge signatures generated by distinct ring members are statistically indistinguishable. The core of the proof relies on the statistical properties of Gaussian sampling over lattices, rather than a reduction to any computational hardness problem.

\begin{theorem}\label{anonymity}
The $\mathcal{LHRS}$ scheme constructed in Section 4 satisfies anonymity under full key exposure.
\end{theorem}

\begin{proof}

    Suppose there exists an adaptive adversary $\mathcal{A}_1$ that attacks the proposed LHRS scheme in accordance with the definition of anonymity under full key exposure. We then construct a polynomial-time algorithm $\mathcal{C}_1$ to simulate the attack environment for $\mathcal{A}_1$. 
    \begin{enumerate}
        \item  \textsf{Initialization}: The challenger $\mathcal{C}_1$ runs the $\mathsf{Setup}(1^{\lambda},  \mathsf{pp})$ algorithm $q_{E}$ times to generate key pairs $(\textsf{pk}_1,  \textsf{sk}_1), \dots, (\textsf{pk}_{q_{E}},  \textsf{sk}_{q_{E}})$, where $\textsf{pk}_{i}=\mathbf{A}_{i}$ and $\textsf{sk}_{i}=\mathbf{T}_{i}$. The public parameter $\mathsf{pp}$ is the same as defined in Section 4 of the proposed scheme. The challenger $\mathcal{C}_1$ stores the tuple $( s,  \textsf{pk}_{s},  \textsf{sk}_{s} )$ (where $1 \leq s \leq q_{E}$) in the list $L_1$, and provides $\{\textsf{pk}_{1},  \textsf{pk}_{2},  \ldots,  \textsf{pk}_{q_{E}}\}$ to the adversary $\mathcal{A}_1$. 

        \item  \textsf{Query Phase}: The adversary $\mathcal{A}_{1}$ makes the following two types of queries: 

        -\textsf{Corruption Query}: The adversary $\mathcal{A}_{1}$ selects a user $s \in \{1,  2,  \dots,  q_E\}$ and sends it to the challenger $\mathcal{C}_1$. $\mathcal{C}_1$ queries the list $L_1$ to retrieve the tuple $(s,  \textsf{pk}_s,  \textsf{sk}_s)$, and returns $\textsf{sk}_s$ to $\mathcal{A}_1$. 

        -\textsf{Signature Query}: The adversary $\mathcal{A}_1$ selects and sends the tuple $(s,  R_\ell,  V_{si})$ to $\mathcal{C}_1$, where the basis vectors of the message subspace $V_{si}$ are $\{\boldsymbol{m}_{si}^{(j)}\}_{j=1}^{k_{0}}$ and it satisfies $A_s \in R_\ell$, with $1 \leq \ell \leq q_E$, $1 \leq i \leq q_s$, and $q_s$ denoting the maximum number of queries that user $s$ can make for the message subspace $V_{si}$ under the ring $R_\ell$. $\mathcal{C}_{1}$ first generates a random tag $\tau_{si} \stackrel{\$}{\leftarrow} \{0, 1\}^{n}$ for $V_{si}$, then generates the signature $\sigma_{\tau_{si}}^{(j)}=( \boldsymbol{e}_{si}^{(j)}, L_{R_{\ell}})$ following the steps of the real signature algorithm. 

        The challenger $\mathcal{C}_1$ returns the tuple $( \boldsymbol{m}_{si}^{(j)},  \sigma_{\tau_{si}}^{(j)})$ to the adversary $\mathcal{A}_1$, where $\textsf{pk}_s \in R_{\ell}$, $\ell \in [q_{E}]$, $1 \leq i \leq q_{s}$, and $1 \leq j \leq k_{0}$. 

        \item[3. ]  \textsf{Challenge Phase}: Finally, the adversary $\mathcal{A}_1$ outputs a challenge tuple $$(s_{0},  s_{1},  R_{\ell^{*}},  \tau^{*},  \boldsymbol{m}^{*}), $$ where $s_{0}$ and $s_{1}$ are indices satisfying $\textsf{pk}_{s_{0}} \in R_{\ell^{*}}$ and $\textsf{pk}_{s_{1}} \in  R_{\ell^{*}}$, and $(\tau^{*},  \boldsymbol{m}^{*})$ is the data message to be signed for the ring $R_{\ell^{*}}$. $\mathcal{C}_1$ randomly selects a bit $b \in \{0, 1\}$, generates the signature $\sigma_{b}^{*}=( \boldsymbol{e}_{b}^{*}, L_{R_{\ell^*}})$ following the same steps as the signature query, and provides $\sigma_{b}^{*}$ to $\mathcal{A}_1$. Finally, $\mathcal{A}_1$ outputs a bit $b^{\prime}$. 
    \end{enumerate}

    Since all responses in the environment simulated by $\mathcal{C}_1$ are generated according to the real algorithm, from the perspective of the adversary $A_{1}$, the behavior of $\mathcal{C}_1$ is exactly that of a real anonymity security experiment.  Furthermore, note that $$V\geq\omega(\sqrt{\log q_{E}n})\geq\omega(\sqrt{\log \ell^{*}n})\geq\eta_{\epsilon}(\Lambda_{q}^{\bot}(A_{R_{\ell^*}})), $$
    regardless of the value of $b$, by Lemma \ref{t3.33}, the distribution of $\boldsymbol{e}_{b}^{*}=\boldsymbol{e}(\boldsymbol{u}_{b}^{*})+\boldsymbol{e}(\boldsymbol{v}_{b}^{*})$ is statistically close to $$\mathcal{D}_{\Lambda_{q}^{\boldsymbol{t}}(A_{R_{\ell^*}}), \sqrt{2}V},$$ where $\boldsymbol{t}= \sum_{j=1}^{k}m_{j}^{*}\bm{\alpha}_{j}$. Therefore, they are computationally indistinguishable. 

    It follows that any polynomial-time adversary $\mathcal{A}_{1}$ cannot distinguish between $\boldsymbol{e}_{0}^{*}$ and $\boldsymbol{e}_{1}^{*}$ with non-negligible probability, i.e.,
    $$\mathbf{Adv}_{LHRS}^{\mathrm{ANON}}(\mathcal{A}_1) = \left|\Pr[b' = b] - \frac{1}{2}\right|=\mathrm{negl}(n). $$
    Thus, the scheme satisfies anonymity. This completes the proof. 
\end{proof}

\subsection{\textbf{Unforgeability}}

Next, we demonstrate that in schemes with smaller ring sizes (for example, less than or equal to $\mathcal{O}(\log n)$), the aforementioned $\mathcal{LHRS}$ satisfies unforgeability.

\vspace{1\baselineskip}	

\begin{theorem}\label{unforgeability}
    \normalfont
    Suppose that $\mathsf{SIS}_{q, h, \ell n, \beta}$ is hard, where $\ell$ denotes the size of the challenge ring and $\beta=2kV\sqrt{2k\ell n}$. Then the $\mathcal{LHRS}$ scheme constructed in this paper achieves unforgeability against internal corruption. 
    
    More specifically, for a polynomial-time adversary $\mathcal{A}_2$ that is allowed to access rings of maximum size $q_{E} (\leq \mathcal{O}(\log n))$ and where each user $s$ can make at most $q_{s}$ signature queries (with respect to message subspaces) in any ring, a polynomial-time algorithm $\mathcal{C}_{2}$ can be constructed that solves the $\mathsf{SIS}_{q, h, \ell n, \beta}$ problem with the following advantage:
    
    $$
    \mathbf{Adv}_{q, h, \ell n, 2V}^{\mathrm{SIS}}(\mathcal{C}_{2}) \geq \frac{\mathbf{Adv}_{\mathrm{LHRS}}^{\mathrm{UNF}}(\mathcal{A}_{2})}{2q_{E} \binom{q_{E}}{q_{E}/2}}-\mathrm{negl}(n). 
    $$
\end{theorem}

\begin{proof}
    Suppose the challenger $\mathcal{C}_{2}$ receives a challenge instance $\mathbf{A}_{R_{\boldsymbol{w}}} \in \mathbb{Z}_{q}^{h \times \ell n}$ of the $\mathsf{SIS}_{q, h, \ell n, \beta}$ problem, whose goal is to find a non-zero short vector $\mathbf{e}_{0}$ such that $\mathbf{A}_{R_{\boldsymbol{w}}}\boldsymbol{e}_{0} \equiv \mathbf{0} \pmod{q}$ and $\|\boldsymbol{e}_{0}\| \leq \beta$. 
    
    The challenger $\mathcal{C}_{2}$ knows that the adversary $\mathcal{A}_{2}$ can perform at most $q_{E}$ key generation operations (i.e., there are at most $q_{E}$ users in the system), and each user $s$ can make at most $q_{s}$ signature queries (relative to the number of message subspaces). $\mathcal{C}_{2}$ attempts to guess which ring $R^{*}$ the adversary $\mathcal{A}_{2}$ will ultimately choose for forgery and its size $\ell$. It randomly selects an $\ell \in [q_{E}]$ (guessing the ring size) and an index vector $\boldsymbol{w}=(w_1, \dots, w_{\ell})\in [q_{E}]^{\ell}$ (guessing the ring member indices), and sets the target ring as $R_{\boldsymbol{w}}=\{\mathbf{A}_{w_1}, \dots, \mathbf{A}_{w_{\ell}}\}$. 
    
    \begin{enumerate}
        \item[$\bullet$] \textsf{Initialization}: In the simulated system, the sizes of parameters $q, k_{0}, k, n, h, V$ follow the same configuration as specified for the proposed scheme in Section 6.3.1. $\mathcal{C}_{2}$ generates public/private key pairs for all $q_{E}$ users as follows:
        
        \begin{enumerate}
            \item [] Case 1: $s\notin \boldsymbol{w}$. $\mathcal{C}_2$ invokes the algorithm $\textsf{TrapGen}(q, h, n)$ to generate a trapdoor matrix pair $(\mathbf{A}_s, \mathbf{T}_s)$ for user $s$. Let $\textsf{pk}_{s} = \mathbf{A}_s$ and $\textsf{sk}_{s} = \mathbf{T}_s$. Store the tuple $(s, \textsf{pk}_s, \textsf{sk}_s)$ in the initially empty list $L_1$. 
            
            \item [] Case 2: $s\in \boldsymbol{w}$.  
            Assume $s = w_{s^{*}}$, i.e., user $s$ is the $s^{*}$-th member of the target ring $R_{\boldsymbol{w}}$. Represent the matrix $\mathbf{A}_{R_{\boldsymbol{w}}}$ in block-column form as $\mathbf{A}_{R_{\mathbf{w}}}=[\mathbf{A}_{w_{1}}|\mathbf{A}_{w_{2}}|\dots|\mathbf{A}_{w_{\ell}}]$. $\mathcal{C}_{2}$ assigns the $s^{*}$-th block $\mathbf{A}_{w_{s^*}}$ of the SIS instance matrix $\mathbf{A}_{R_{\boldsymbol{w}}}$ to user $s$, and sets its public key as $\textsf{pk}_{s}= \mathbf{A}_{s}= \mathbf{A}_{w_{s^*}}$. Finally, $\mathcal{C}_{2}$ also stores $(s, \textsf{pk}_{s})$ in the list $L_1$. 
        \end{enumerate}
        
        The challenger $\mathcal{C}_{2}$ sends the public keys of all users $L = \{\textsf{pk}_{1}, \ldots, \textsf{pk}_{q_{E}}\}$ to the adversary $\mathcal{A}_{2}$. 	
        
        \item [$\bullet$] \textsf{Query Phase}: $\mathcal{A}_{2}$ makes the following two types of queries:
        
        -\textsf{Private Key (Corruption) Query}: When $\mathcal{A}_{2}$ queries the private key of user $s$, $\mathcal{C}_{2}$ responds as follows:
        
        \begin{enumerate}
            \item [] Case 1: $s\notin \boldsymbol{w}$. $\mathcal{C}_{2}$ looks up the entry $(s, \textsf{pk}_{s}, \textsf{sk}_{s})$ in list $L_1$. If the corresponding entry exists, it sends the corresponding $\textsf{sk}_{s}$ to $\mathcal{A}_{2}$; if not, $\mathcal{C}_{2}$ invokes the algorithm $\textsf{TrapGen}(q, h, n)$ to generate a trapdoor pair $(\mathbf{A}_{s}, \mathbf{T}_{s})$, sets $\textsf{pk}_{s} = \mathbf{A}_s$ and $\textsf{sk}_{s} = \mathbf{T}_s$, returns $\textsf{sk}_{s}$ to the adversary $\mathcal{A}_{2}$, and stores this entry $(s, \textsf{pk}_{s}, \textsf{sk}_{s})$ in list $L_1$.
            
            \item [] Case 2: $s\in \boldsymbol{w}$. $\mathcal{C}_{2}$ aborts. 
        \end{enumerate}
        
        -\textsf{Hash Query}: When $\mathcal{A}_{2}$ queries the hash value of a uniformly chosen tag $\tau_{si}$ for the message subspace $V_{si}$, where $V_{si}$ is generated by the basis vectors $\{\boldsymbol{m}_{si}^{(j)}\}_{j=1}^{k_0}$ with $1\leq i\leq q_{s}$, $\mathcal{C}_{2}$ checks list $L_{2}$. If the entry $(\tau_{si},\{ \boldsymbol{h}_{si}^{(j)}, \bm{\alpha}_{si}^{(j)} \}_{j=1}^{k})$ exists, $\mathcal{C}_{2}$ returns $\bm{\alpha}_{si}^{(1)},\dots,\bm{\alpha}_{si}^{(k)}$ to the adversary as the response to the tag $\tau_{si}$; otherwise, it performs the following computations:
        
        \begin{enumerate}
            \item[(1)] Compute $ \boldsymbol{h}_{si}^{(\ell)}\leftarrow \textbf{SampleDom}(1^{\ell n},s)$, where $s=V/\sqrt{k/2}$.
            
            \item[(2)] Compute $\bm{\alpha}_{si}^{(j)}=\mathbf{A}_{R_{\boldsymbol{w}}}\boldsymbol{h}_{si}^{(j)}\pmod{q}$ for $\ell\in[k]$. (By Lemma \ref{l3.4}, the distribution of $\bm{\alpha}_{si}^{(j)}$ is indistinguishable from the uniform distribution over $\mathbb{Z}_{q}^{h}$ at this point.)
        \end{enumerate}
        
        At this point, return $\bm{\alpha}_{si}^{(1)},\dots,\bm{\alpha}_{si}^{(k)}$ to the adversary, and store $(\tau_{si},\{ \boldsymbol{h}_{si}^{(j)}, \bm{\alpha}_{si}^{(j)} \}_{j=1}^{k})$ in list $L_{2}$.

        -\textsf{Signature Query}: $\mathcal{A}_{2}$ selects the tuple $(s, R_{\boldsymbol{w}^*} , V_{si}, \tau_{si} )$ and sends it to $\mathcal{C}_{2}$. $\mathcal{C}_{2}$ checks list $L_{3}$ to find the entry $(s, R_{\boldsymbol{w}^{*}}, V_{si}, \tau_{si} , \{\boldsymbol{m}_{si}^{(j)}, \sigma_{si}^{(j)}\}_{j=1}^{k_{0}})$. If it exists, it returns the corresponding entries $\sigma_{si}^{(1)},\dots, \sigma_{si}^{(k_{0})}$ to the adversary; otherwise, it computes the signatures $\{\sigma_{si}^{(j)}\}_{j=1}^{k_{0}}$ according to the following cases and stores the new tuple in list $L_{3}$:
        
        \begin{enumerate}
            \item [] Case 1: $R_{\boldsymbol{w}^*}=R_{\boldsymbol{w}}$.  
            \begin{enumerate}
                \item [(1)] Compute $(\boldsymbol{u}_{si}^{(j)}, \boldsymbol{v}_{si}^{(j)}){\leftarrow}\textbf{Decompose}(\boldsymbol{m}_{si}^{(j)})$, and let 
                $$ \boldsymbol{u}_{si}^{(j)}=(u_{si}^{(j1)}, \dots, u_{si}^{(jk)}),  $$
                $$ \boldsymbol{v}_{si}^{(j)}=(v_{si}^{(j1)}, \dots, v_{si}^{(jk)}). $$
                \item [(2)] Compute 
                $$\boldsymbol{e}(\boldsymbol{u}_{si}^{(j)})=\sum_{p=1}^{k}u_{si}^{(jp)}\boldsymbol{h}_{si}^{(j)} $$ 
                and 
                $$\boldsymbol{e}(\boldsymbol{v}_{si}^{(j)})=\sum_{p=1}^{k}v_{si}^{(jp)}\boldsymbol{h}_{si}^{(j)} . $$ 
                Set $\boldsymbol{e}_{si}^{(j)}=\boldsymbol{e}(\boldsymbol{u}_{si}^{(j)})+\boldsymbol{e}(\boldsymbol{v}_{si}^{(j)})$ and $\sigma_{si}^{(j)}=( \boldsymbol{e}_{si}^{(j)}, L_{R_{\boldsymbol{w}^{*}}})$. 
            \end{enumerate}
            
            \vspace{1\baselineskip}
            \item [] Case 2: $R_{\boldsymbol{w}^*}\neq  R_{\boldsymbol{w}}$ and $ s \notin \boldsymbol{w}$. 
            \begin{enumerate}
                \item [(1)] $\mathcal{C}_{2}$ retrieves the tuple $(s, \textsf{pk}_{s}, \textsf{sk}_{s})$ from list $L_{2}$ to obtain the private key $\mathbf{T}_{s} $ of user $s$. 
                \item [(2)] Compute $(\boldsymbol{u}_{si}^{(j)}, \boldsymbol{v}_{si}^{(j)}){\leftarrow}\textsf{Decompose}(\boldsymbol{m}_{si}^{(j)})$, and let 
                $$ \boldsymbol{u}_{si}^{(j)}=(u_{si}^{(j1)}, \dots, u_{si}^{(jk)}),  $$
                $$ \boldsymbol{v}_{si}^{(j)}=(v_{si}^{(j1)}, \dots, v_{si}^{(jk)}). $$
                \item [(3)] Compute 
                $$\boldsymbol{t}(\boldsymbol{u}_{si}^{(j)})=\sum_{p=1}^{k}u_{si}^{(jp)}\bm{\alpha}_{si}^{(p)}$$
                and 
                $$\boldsymbol{t}(\boldsymbol{v}_{si}^{(j)})=\sum_{p=1}^{k}v_{si}^{(jp)}\bm{\alpha}_{si}^{(p)}. $$
                \item [(4)] Compute 
                $$ \mathbf{e}(\boldsymbol{u}_{si}^{(j)})\leftarrow \textsf{GenSamplePre}(\mathbf{A}_{R_{\boldsymbol{w}^{*}}}, \mathbf{A}_{s}, \mathbf{T}_{s}, \boldsymbol{t}(\boldsymbol{u}_{si}^{(j)}), V ) $$
                and 
                $$ \mathbf{e}(\boldsymbol{v}_{si}^{(j)})\leftarrow \textsf{GenSamplePre}(\mathbf{A}_{R_{\boldsymbol{w}^{*}}}, \mathbf{A}_{s}, \mathbf{T}_{s}, \boldsymbol{t}(\boldsymbol{v}_{si}^{(j)}), V ) . $$
                Set $\boldsymbol{e}_{si}^{(j)}=\boldsymbol{e}(\boldsymbol{u}_{si}^{(j)})+\boldsymbol{e}(\boldsymbol{v}_{si}^{(j)})$ and $\sigma_{si}^{(j)}=( \boldsymbol{e}_{si}^{(j)}, L_{R_{\boldsymbol{w}^{*}}})$. 
            \end{enumerate}
            
            \vspace{1\baselineskip}
            \item [] Case 3: $R_{\boldsymbol{w}^*}\neq  R_{\boldsymbol{w}}$ and $ s \in\boldsymbol{w}$. 
            \begin{enumerate}
                \item [(1)] $\mathcal{C}_{2}$ finds a member $s^{*}$ in the ring $R_{\boldsymbol{w}^{*}}$ such that $s^{*} \notin \boldsymbol{w}$, then retrieves the tuple $(s^{*}, \textsf{pk}_{s^{*}}, \textsf{sk}_{s^{*}})$ from list $L_{2}$ to obtain the private key $\mathbf{T}_{s^{*}}$ of user $s^{*}$. 
                \item [(2)] The remaining steps follow the signature generation process in Case 2 to produce the signature $\sigma_{si}^{(j)}$. (Note that any member of the ring $R_{\boldsymbol{w}^*}$ can generate a valid signature on behalf of the ring.) 
            \end{enumerate}
        \end{enumerate}	
    \end{enumerate}
    
    We now prove that the signatures provided by $\mathcal{C}_{2}$ are both valid and indistinguishable from the signatures in the real scheme.   
    
    Since the signatures in Case 2 and Case 3 are generated by the real signature algorithm, it suffices to prove that the signatures generated in Case 1 are both valid and computationally indistinguishable from the signatures in the real scheme. Furthermore, since $L_{R_{w^*}}$ in the signature component is an index vector, we only need to verify that the signature component $\boldsymbol{e}_{si}^{(j)}$ is computationally indistinguishable from the corresponding signature component in the real signature scheme. 
    
    By Lemma \ref{t3.33} and Lemma \ref{co}, the distributions of $\boldsymbol{e}(\boldsymbol{u}_{si}^{(j)})$ and $\boldsymbol{e}(\boldsymbol{v}_{si}^{(j)})$ are both statistically close to the distribution $\mathcal{D}_{\mathbb{Z}^{\ell^{*} n}, V}$, where $\ell^{*}$ denotes the number of ring members in $R_{\boldsymbol{w}^{*}}$. 
    
    Additionally, since $\mathbf{A}_{R_{\boldsymbol{w}^{*}}}\cdot \boldsymbol{e}(\boldsymbol{u}_{si}^{(j)}) = \boldsymbol{t}(\boldsymbol{u}_{si}^{(j)}) \pmod{q}$, where $\boldsymbol{t}(\boldsymbol{u}_{si}^{(j)}) = \sum_{p=1}^{k}u_{si}^{(jp)} \bm{\alpha}_{si}^{(p)}$, by Lemma \ref{l3.777}, $\boldsymbol{e}(\boldsymbol{u}_{si}^{(j)})$ is statistically close to the distribution $\mathcal{D}_{\mathcal{L}_{1}, V}(\mathbf{A}_{R_{\mathbf{w}^{*}}})$, where $\mathcal{L}_{1}=\Lambda_{q}^{\boldsymbol{t}(\boldsymbol{u}_{si}^{(j)})}$. Similarly, $\boldsymbol{e}(\boldsymbol{v}_{si}^{(j)})$ is statistically close to the distribution $\mathcal{D}_{\mathcal{L}_{2}, V}(\mathbf{A}_{R_{\boldsymbol{w}^{*}}})$, where $\mathcal{L}_{2}=\Lambda_{q}^{\boldsymbol{t}(\boldsymbol{v}_{si}^{(j)})}$. 
    
    By Lemma \ref{lem5}, we have
    \[
    \left\Vert \boldsymbol{e}(\boldsymbol{u}_{si}^{(j)}) \right\Vert = \left\Vert \sum_{p=1}^{k}u_{si}^{(jp)} \boldsymbol{h}_{si}^{(p)} \right\Vert \leq k \max_{1 \leq p \leq k} \left\Vert\boldsymbol{h}_{si}^{(p)} \right\Vert \leq k \frac{V}{\sqrt{k/2}} \sqrt{\ell^{*} n} = V \sqrt{2k \ell^{*} n}. 
    \]
    
    Similarly, 
    \[
    \left\Vert \boldsymbol{e}(\boldsymbol{v}_{si}^{(j)}) \right\Vert \leq V \sqrt{2k \ell^* n}. 
    \]
    
    Therefore, 
    \[
    \left\Vert \boldsymbol{e}_{si}^{(j)} \right\Vert = \left\Vert \boldsymbol{e}(\boldsymbol{u}_{si}^{(j)}) + \boldsymbol{e}(\boldsymbol{v}_{si}^{(j)}) \right\Vert \leq 2V \sqrt{2k \ell^* n} \leq k V \sqrt{2k \ell^* n}. 
    \]
    
    Meanwhile, the adversary $\mathcal{A}_{2}$ can verify that the following equality holds:
    \begin{align*}
        \mathbf{A}_{R_{\boldsymbol{w}^*}} \cdot \boldsymbol{e}_{si}^{(j)} \pmod{q}
        &= \mathbf{A}_{R_{\boldsymbol{w}^*}} \cdot \left(\boldsymbol{e}(\boldsymbol{u}_{si}^{(j)}) + \mathbf{e}(\boldsymbol{v}_{si}^{(j)})\right) \pmod{q}\\
        &= \mathbf{A}_{R_{\boldsymbol{w}^*}} \left(\sum_{p=1}^{k}  u_{si}^{(jp)} \boldsymbol{h}_{si}^{(p)} + \sum_{p=1}^{k} v_{si}^{(jp)} \boldsymbol{h}_{si}^{(p)}\right) \pmod{q}\\
        &= \sum_{p=1}^{k} u_{si}^{(jp)} \mathbf{A}_{R_{\boldsymbol{w}^*}} \cdot \boldsymbol{h}_{si}^{(p)} + \sum_{p=1}^{k}v_{si}^{(jp)} \mathbf{A}_{R_{\boldsymbol{w}^*}} \cdot \boldsymbol{h}_{si}^{(p)} \pmod{q}\\
        &= \sum_{p=1}^{k} u_{si}^{(jp)} \bm{\alpha}_{si}^{(p)}+ \sum_{p=1}^{k}  v_{si}^{(jp)} \bm{\alpha}_{si}^{(p)} \\
        &= \sum_{p=1}^{k} \left(u_{si}^{(jp)} + v_{si}^{(jp)}\right) \bm{\alpha}_{si}^{(p)} \\
        &= \sum_{p=1}^{k}m_{si}^{(jp)} \bm{\alpha}_{si}^{(p)} \\
        &= \boldsymbol{t}_{si}^{(j)} 
    \end{align*}
    
    Thus, the signatures simulated by $\mathcal{C}_{2}$ are both valid and computationally indistinguishable from the signatures in the real scheme. 
    
    After completing the queries, $\mathcal{A}_{2}$ outputs a valid forgery $(R^{*}, \tau^{*}, \boldsymbol{m}^{*}, \sigma^{*})$, where $\sigma^{*} = (\boldsymbol{e}^{*}, L_{R^{*}})$ and $L_{R^{*}}$ denotes the index of the ring $R^{*}$. 
    
    If $R^{*} \neq R_{\boldsymbol{w}}$, then $\mathcal{C}_{2}$ aborts the simulation. Otherwise, it handles the following two cases:
    
    \begin{enumerate}
        \item [(1)] Assume the adversary performs a Type I forgery: i.e., $\tau^*\neq \tau_{si}$ where $s\in [q_{E}]$ and $i\in[q_{s}]$.
        
        The challenger $\mathcal{C}_{2}$ then checks list $L_{2}$ to find $\tau^{*}$ and its associated record $(\tau^*,\{ \boldsymbol{h}_{j}^{*}, \bm{\alpha}_{j}^{*} \}_{j=1}^{k})$ (if it does not exist, simply perform the hash query again). Let $\boldsymbol{h}^{*}=\sum_{j=1}^{k}m_{j}^{*}\boldsymbol{h}_{j}^{*}$; it can be verified that $(\boldsymbol{h}^{*},L_{R^{*}} )$ is also a signature for $\boldsymbol{m}^*$, thus
        $$  \mathbf{A}_{R_{\boldsymbol{w}}} \cdot \boldsymbol{e}^{*}=   \mathbf{A}_{R_{\boldsymbol{w}}} \cdot \boldsymbol{h}^{*}\pmod{q}, $$
        and therefore $ \mathbf{A}_{R_{\boldsymbol{w}}} (\boldsymbol{e}^{*}-\boldsymbol{h}^{*})=0\pmod{q}.$ Furthermore, since:
        \[
        \left\Vert  \boldsymbol{e}^* -\boldsymbol{h}^*\right\Vert \leq 2kV\sqrt{2k\ell n}.
        \]
        When $\boldsymbol{e}^{*}-\boldsymbol{h}^{*}\neq 0 $, $\mathcal{C}_{2}$ outputs $\boldsymbol{e}^{*}-\boldsymbol{h}^{*}$ as a solution to $\mathsf{SIS}_{q, h, \ell n, \beta}$. Additionally, since $\boldsymbol{h}^{*}$ is the result of Gaussian sampling with a conditional min-entropy of at least $\omega(\log n)$, we have
        $$    \Pr\left[\boldsymbol{e}^* = \boldsymbol{h}^*\right] \leq 2^{-\omega(\log n)} \leq \mathsf{negl}(n).$$
        
        \item [(2)] Assume the adversary performs a Type II forgery: i.e., there exist $s,i$ such that $\tau^*=\tau_{si}$ where $s\in [q_{E}]$ and $i\in[q_{s}]$.		
        
        The challenger $\mathcal{C}_{2}$ checks list $L_{2}$ to find $\tau_{si}$ and its associated record $(\tau_{si},\{ \boldsymbol{h}_{si}^{(j)}, \bm{\alpha}_{si}^{(j)} \}_{j=1}^{k})$. Let $\boldsymbol{h}^{(*)}=\sum_{j=1}^{k}m_{j}^{*}\boldsymbol{h}_{si}^{(j)}$; it can be verified that $(\boldsymbol{h}^{(*)},L_{R^{*}} )$ is also a signature for $\boldsymbol{m}^*$. Similar to the discussion of Type I forgery, $\boldsymbol{e}^{*}-\boldsymbol{h}^{*}$ can still be output as a solution to $\mathsf{SIS}_{q, h, \ell n, \beta}$.		
    \end{enumerate}
    
    The probability that $\mathcal{C}_{2}$ does not abort in the game is at least $\frac{1}{2q_{E} \binom{q_{E}}{q_{E}/2}}. $ 
    
    Therefore, combining both cases, $\mathcal{C}_{2}$ will obtain a non-zero short solution to the $\mathsf{SIS}_{q, h, \ell n, \beta}$ problem with the following advantage:
    $$
    \mathbf{Adv}_{q, h, \ell n, 2V}^{\mathrm{SIS}}(\mathcal{C}_{2}) \geq \frac{\mathbf{Adv}_{\mathrm{LHRS}}^{\mathrm{UNF}}(\mathcal{A}_{2})}{2q_{E} \binom{q_{E}}{q_{E}/2}}-\mathrm{negl}(n). 
    $$
\end{proof}

\textbf{Analysis of the Security-Loss Factor.} The assumption $q_{E} \leq \mathcal{O}(\log n)$ implies that the binomial coefficient $\binom{q_{E}}{q_{E}/2}$ is at most polynomial in $n$. Specifically, using the approximation $\binom{2k}{k} \sim 4^{k} / \sqrt{\pi k}$, the entire denominator $2q_{E} \cdot \binom{q_{E}}{q_{E}/2}$ is of order $\mathcal{O}(\mathsf{poly}(n))$. Therefore, the reduction only incurs a polynomial loss in the advantage. A non-negligible success probability for the forger $\mathcal{A}_{2}$ thus translates to a non-negligible advantage for the SIS solver $\mathcal{B}_{2}$.

\vspace{1\baselineskip}
\textbf{Worst-case to Average-case Reduction.} Building upon Theorem \ref{t3.1} in \cite{37}, the computational complexity of the $\mathrm{SIS}_{q,n,m,\beta}$ can be shown to be polynomially equivalent to worst-case approximations of the Shortest Independent Vectors Problem (SIVP), given that the modulus satisfies $q \geq \beta \cdot \omega(\sqrt{n \log n})$. Under this condition, the approximation factor for SIVP is on the order of $\beta \cdot \tilde{O}(\sqrt{n})$. In our construction, we require that \color{blue}\( q \geq (nk)^3 \).\color{black} Therefore, there exists a sufficiently large \( q \) that meets the conditions of Theorem \ref{t3.1}, thereby ensuring its applicability.

\subsection{Comparison with the CK17 Scheme}

We compare our proposed linearly homomorphic ring signature scheme  with the CK17 scheme proposed by Choi and Kim \cite{41, 42} (see also the SCIS18 version).  The comparison focuses on four key aspects: (1) signature size and asymptotic efficiency; (2) practical constraints (ring size limit) and core security properties (anonymity, unforgeability); (3) cryptographic assumptions (hard problem, quantum resistance); and (4) the completeness of the formal security proof.

\begin{table}[htbp]
\centering
\caption{Comparison of our scheme with the CK17 scheme.}
\label{tab:comparison}
\begin{tabular}{|l|c|c|}
\hline
\textbf{Feature} & \textbf{Our Scheme} & \textbf{CK17 Scheme} \\
\hline
Signature size & \(\mathcal{O}(\ell n)\) & \(\mathcal{O}((\ell+1)n)\) \\
Ring size limit & \(\mathcal{O}(\log n)\) & No explicit bound \\
Anonymity &   Yes  &   Yes \\
Unforgeability & Adaptive & Non-adaptive \\
Hard problem & \(\mathrm{SIS}\)  & \(\mathrm{SIS}\) \\
Quantum-resistant & Yes & Yes \\
Full security proof & Yes & No \\
\hline
\end{tabular}
\end{table}

As summarized in Table~\ref{tab:comparison}, our scheme outperforms the CK17 scheme in several aspects, but it also has certain limitations. First, in terms of signature length, our scheme reduces the asymptotic complexity from \(\mathcal{O}((\ell+1)n)\) to \(\mathcal{O}(\ell n)\), where \(\ell\) represents the number of ring members. 

Second, a key limitation of our scheme is the explicit restriction of the ring size to \(\mathcal{O}(\log n)\). This constraint arises from the security reduction technique: to achieve a polynomial loss factor in the unforgeability proof, the binomial coefficient \(\binom{q_E}{q_E/2}\) must be bounded by a polynomial in \(n\), which requires \(q_E \leq \mathcal{O}(\log n)\). In contrast, the CK17 scheme imposes no such restriction, offering greater flexibility in practice. However, this flexibility comes at the cost of security: due to the lack of a formal reduction, the security of the CK17 scheme against adaptive attacks for rings of arbitrary size is not clearly established.

Third, both schemes base their anonymity on the statistical indistinguishability of Gaussian samples. In our scheme, this is formally established in Theorem~\ref{anonymity}, while the CK17 scheme relies on a similar heuristic argument. Regarding unforgeability, our proof considers adversaries capable of corrupting internal users (full key exposure) and adaptively issuing signing queries, achieving a tight reduction to the SIS problem. The CK17 scheme, however, only provides a non-adaptive security sketch without a complete reduction, leaving its actual security level unverified.

Fourth and most importantly, the key distinction lies in the completeness of the security proof. Although CK17 pioneered the concept of linear homomorphic ring signatures and formally defined the adversary model, it only provides a proof sketch (see Theorems 3, 4, and 5 in the SCIS18 version) and falls short of a full reduction to a hard lattice problem. In contrast, this work presents the first complete security proof for a linear homomorphic ring signature scheme, including: (i) a tight reduction from the SIS problem to unforgeability under insider corruption (Theorem~\ref{unforgeability}), and (ii) a rigorous statistical argument for anonymity (Theorem~\ref{anonymity}).

In summary, our scheme prioritizes provable security and tighter reductions at the expense of ring size flexibility, whereas the CK17 scheme offers enhanced practical flexibility in ring size but lacks a complete security argument. This trade-off highlights the inherent challenges in designing homomorphic ring signature schemes and points to future research directions: achieving provable security while supporting larger or even unbounded ring sizes.

\section{Conclusions}

Existing research offers limited exploration on the integration of ring signatures and homomorphic signatures. The concept of linearly homomorphic ring signatures has only been mentioned without providing a complete and provably secure construction. Moreover, its practical design faces the challenge of balancing security, efficiency, and functional completeness. To address this issue, this paper proposes the first provably secure lattice-based linearly homomorphic ring signature scheme, which retains the core security properties of both primitives and simultaneously achieves quantum-resistant security.

Several directions remain for further investigation and improvement of the proposed scheme:
\begin{enumerate}
    \item[(1)] Construct new schemes with short signatures (whose size is independent of the ring size) while maintaining reasonable computational efficiency.
    
 \item[(2)]To design provably secure homomorphic ring signatures that support larger or even unbounded ring sizes;
    
    \item[(3)] Study the feasibility of designing homomorphic ring signature schemes that support polynomial or fully homomorphic computation.
    \item[(4)] Explore whether homomorphic ring signature schemes with tight or almost tight security reductions can be constructed under standard lattice assumptions.
\end{enumerate}

\section*{Author contributions}
These authors contributed equally to this work.

	 \section*{Availability of data and materials}
No data were used in the present study. Thus, there is no relevant data to share or report regarding availability.

\section*{Declarations}
The authors declare that they have no known competing financial interests or personal relationships that could have influenced the work reported in this paper. All authors confirm that the research complies with ethical standards and is free from conflicts of interest.


\begin{thebibliography}{99}
	\bibitem{1}
Diffie W, Hellman M.E.
\newblock New directions  in cryptography[J].
\newblock IEEE Transactions on Information Theory, 1976,  22(6): 644-654.

\bibitem{2}
Liu F, Zheng Z, Gong Z, et al. A survey on lattice-based digital signature[J]. Cybersecurity, 2024, 7(1): 7.


\bibitem{3}
Lyubashevsky V.
\newblock Lattice signatures without trapdoors[C].
\newblock In Annual International Conference on the Theory and Applications of Cryptographic Techniques. Berlin, Heidelberg: Springer Berlin Heidelberg, 2012: 738-755.	

\bibitem{4}
Katz, J. 
\newblock  Digital signatures[M].
\newblock  Mathematische Annalen,  Berlin: Springer. 2010. \url{https://doi.org/10.1007/978-0-387-27712-7} 

\bibitem{5}
Boneh D, Boyen X, Shacham H.
\newblock Short group signatures[C]. 
\newblock  In Annual international cryptology conference. Berlin, Heidelberg: Springer Berlin Heidelberg, 2004: 41-55.

\bibitem{6}
Pointcheval D, Stern J. 
\newblock Security arguments for digital signatures and blind signatures[J].
\newblock Journal of cryptology, 2000, 13: 361-396.


\bibitem{7}
Rivest RL, Shamir A, Tauman Y. 
\newblock How to leak a secret[C].
\newblock Advances in Cryptology—ASIACRYPT 2001: 7th International Conference on the Theory and Application of Cryptology and Information Security Gold Coast, Australia, December 9–13, 2001 Proceedings 7. Springer Berlin Heidelberg, 2001: 552-565.



\bibitem{8}
Rivest R L, Shamir A, Tauman Y.
\newblock  How to leak a secret: Theory and applications of ring signatures[J].
\newblock Theoretical Computer Science: Essays in Memory of Shimon Even, 2006: 164-186.




\bibitem{9}
Tsang P P, Wei V K.
\newblock  Short linkable ring signatures for e-voting, e-cash and attestation[C].
\newblock In International Conference on Information Security Practice and Experience. Berlin, Heidelberg: Springer Berlin Heidelberg, 2005: 48-60.


\bibitem{10}
Ta A T, Khuc T X, Nguyen T N, et al.
\newblock  Efficient unique ring signature for blockchain privacy protection[C].
\newblock In Information Security and Privacy: 26th Australasian Conference, ACISP 2021, Virtual Event, December 1–3, 2021, Proceedings 26. Springer International Publishing, 2021: 391-407.




\bibitem{11}
Thyagarajan S A K, Malavolta G, Schmid F, et al.
\newblock Verifiable timed linkable ring signatures for scalable payments for monero[C].
\newblock In European Symposium on Research in Computer Security. Cham: Springer Nature Switzerland, 2022: 467-486.




\bibitem{12}
Fujisaki E, Suzuki K. 
\newblock Traceable ring signature[C].
\newblock  In International Workshop on Public Key Cryptography. Berlin, Heidelberg: Springer Berlin Heidelberg, 2007: 181-200.


\bibitem{13}
Bender A, Katz J, Morselli R.
\newblock  Ring signatures: Stronger definitions, and constructions without random oracles[C].
\newblock In Theory of Cryptography Conference. Berlin, Heidelberg: Springer Berlin Heidelberg, 2006: 60-79.



\bibitem{14}
Liu J K, Au M H, Susilo W, et al.
\newblock  Linkable ring signature with unconditional anonymity[J].
\newblock IEEE Transactions on Knowledge and Data Engineering, 2013, 26(1): 157-165.



\bibitem{15}
Wang J, Sun B.
\newblock  Ring signature schemes from lattice basis delegation[C].
\newblock In Information and Communications Security: 13th International Conference, ICICS 2011, Beijing, China, November 23-26, 2011. Proceedings 13. Springer Berlin Heidelberg, 2011: 15-28.		


\bibitem{16}
Shacham H, Waters B.
\newblock Efficient ring signatures without random oracles[C].
\newblock Public Key Cryptography–PKC 2007: 10th International Conference on Practice and Theory in Public-Key Cryptography Beijing, China, April 16-20, 2007. Proceedings 10. Springer Berlin Heidelberg, 2007: 166-180.



	\bibitem{17}
Johnson R, Molnar D, Song D, et al.
\newblock Homomorphic signature schemes[C]. 
\newblock In Cryptographers’ track at the RSA conference. Berlin, Heidelberg: Springer Berlin Heidelberg, 2002: 244-262.


	\bibitem{18}
Yun A, Cheon J H, Kim Y.
\newblock On homomorphic signatures for network coding[J]. 
\newblock IEEE Transactions on Computers, 2010, 59(9): 1295-1296.		 


\bibitem{19}
Attrapadung N, Libert B.
\newblock Homomorphic network coding signatures in the standard model[C] 
\newblock In Public Key Cryptography–PKC 2011: 14th International Conference on Practice and Theory in Public Key Cryptography, Taormina, Italy, March 6-9, 2011. Proceedings 14. Springer Berlin Heidelberg, 2011: 17-34.	 


\bibitem{20}
Wu B, Wang C, Yao H.
\newblock  A certificateless linearly homomorphic signature scheme for network coding and its application in the IoT[J].
\newblock  Peer-to-Peer Networking and Applications, 2021, 14(2): 852-872.


\bibitem{21}
Li J, Zhang Y, Chen X, et al.
\newblock  Secure attribute-based data sharing for resource-limited users in cloud computing[J].
\newblock  computers \& security, 2018, 72: 1-12.

\bibitem{22}
Li P, Li J, Huang Z, et al.
\newblock  Privacy-preserving outsourced classification in cloud computing[J]. 
\newblock  Cluster Computing, 2018, 21: 277-286.

\bibitem{23}
Emmanuel N, Khan A, Alam M, et al. 
\newblock  Structures and data preserving homomorphic signatures[J].
\newblock  Journal of Network and Computer Applications, 2018, 102: 58-70.



 56:	\bibitem{24}
Guo H, Tian K, Liu F, et al.
\newblock  Linearly Homomorphic Signature with Tight Security on Lattice[J].
\newblock   arXiv preprint arXiv:2412.01641, 2024.

\bibitem{25}
Boneh D, Freeman D M. 
\newblock Linearly homomorphic signatures over binary fields and new tools for lattice-based signatures[C].
\newblock In International Workshop on Public Key Cryptography. Berlin, Heidelberg: Springer Berlin Heidelberg, 2011: 1-16.

\bibitem{26}
Chen W, Lei H, Qi K.
\newblock Lattice-based linearly homomorphic signatures in the standard model[J]. 
\newblock Theoretical Computer Science, 2016, 634: 47-54.

\bibitem{27}
Schabhüser L, Buchmann J, Struck P. 
\newblock  A linearly homomorphic signature scheme from weaker assumptions[C].
\newblock  In Cryptography and Coding: 16th IMA International Conference, IMACC 2017, Oxford, UK, December 12-14, 2017, Proceedings 16. Springer International Publishing, 2017: 261-279.


\bibitem{28}
Wang F H, Hu Y P, Wang B C.
\newblock Lattice-based linearly homomorphic signature scheme over binary field[J].
\newblock Science China Information Sciences, 2013, 1-9.


	\bibitem{29}
Boneh D, Freeman D M.
\newblock Homomorphic signatures for polynomial functions[C].
\newblock In Advances in Cryptology–EUROCRYPT 2011: 30th Annual International Conference on the Theory and Applications of Cryptographic Techniques, Tallinn, Estonia, May 15-19, 2011. Proceedings 30. Springer Berlin Heidelberg, 2011: 149-168.

\bibitem{30}
Hiromasa R, Manabe Y, Okamoto T.
\newblock Homomorphic signatures for polynomial functions with shorter signatures[C].
\newblock In The 30th symposium on cryptography and information security, Kyoto. 2013.


\bibitem{31}
Catalano D, Fiore D, Warinschi B. 
\newblock Homomorphic signatures with efficient verification for polynomial functions[C].
\newblock In Annual Cryptology Conference. Berlin, Heidelberg: Springer Berlin Heidelberg, 2014: 371-389.


\bibitem{32}
Arita S, Kozaki S.
\newblock  A homomorphic signature scheme for quadratic polynomials[C].
\newblock In 2017 IEEE International Conference on Smart Computing (SMARTCOMP). IEEE, 2017: 1-6.



	\bibitem{33}
Gorbunov S, Vaikuntanathan V, Wichs D.
\newblock  Leveled fully homomorphic signatures from standard lattices[C].
\newblock In Proceedings of the forty-seventh annual ACM symposium on Theory of computing. 2015: 469-477.

\bibitem{34}
Wang Y, Wang M.
\newblock A new fully homomorphic signatures from standard lattices[C].
\newblock In  International Conference on Wireless Algorithms, Systems, and Applications. Cham: Springer International Publishing, 2020: 494-506.

\bibitem{35}
Boyen X, Fan X, Shi E.
\newblock  Adaptively secure fully homomorphic signatures based on lattices[J].
\newblock Cryptology ePrint Archive, 2014.



  
   
	\bibitem{36}
Zheng Z.
\newblock Modern Cryptography Volume 1: A Classical Introduction to Informational and Mathematical Principle[M].
\newblock Springer Nature, 2022.


\bibitem{37}
Gentry C, Peikert C, Vaikuntanathan V.
\newblock Trapdoors for hard lattices and new cryptographic constructions[C].
\newblock Proceedings of the fortieth annual ACM symposium on Theory of computing. 2008: 197-206.


\bibitem{38}
Agrawal S, Boneh D, Boyen X.
\newblock Lattice basis delegation in fixed dimension and shorter-ciphertext hierarchical IBE[C].
\newblock Advances in Cryptology–CRYPTO 2010: 30th Annual Cryptology Conference, Santa Barbara, CA, USA, August 15-19, 2010. Proceedings 30. Springer Berlin Heidelberg, 2010: 98-115.


	\bibitem{39}
Micciancio D, Regev O.
\newblock Worst-case to average-case reductions based on Gaussian measures[J].
\newblock SIAM journal on computing, 2007, 37(1): 267-302.

\bibitem{40}
Cash D, Hofheinz D, Kiltz E.
\newblock How to delegate a lattice basis[J].
\newblock Cryptology ePrint Archive, 2009.

\bibitem{41}
Choi R, Kim K. Design of new linearly homomorphic signatures on lattice[C].Symposium on Cryptography and Information Security. 2017.

\bibitem{42}
Choi R, Kim K. Revisiting CK17 Linearly Homomorphic Ring Signature based on SIS[C].2018 Symposium on Cryptography and Information Security (SCIS 2018). IEICE Technical Committee on Information Security, 2018.

\end{thebibliography}
\end{document}